\documentclass[11pt,fleqn]{article}
\usepackage{amssymb,latexsym,amsmath,amsfonts,graphicx, ebezier}
\usepackage{pictex,color}
\usepackage{pict2e}

    \advance\textheight by \topskip

    \numberwithin{equation}{section}

\def\beq{\begin{equation}}
\def\eeq{\end{equation}}

\def\bit{\begin{itemize}}
\def\eit{\end{itemize}}

\makeatletter
\def\eqalign#1{\null\vcenter{\def\\{\cr}\openup\jot\m@th
  \ialign{\strut$\displaystyle{##}$\hfil&$\displaystyle{{}##}$\hfil
      \crcr#1\crcr}}\,}
\makeatother

\newcommand{\be}{\begin{equation}}
\newcommand{\ee}{\end{equation}}

    \def\Re{{\rm Re \,}}
    \def\Im{{\rm Im \,}}

    \def\bigO{{\cal O}}

    \def\P2n{{\rm P}_{{\rm II}}^{(n)}}

    \newtheorem{theorem}{Theorem}[section]
    
    \newtheorem{corollary}[theorem]{Corollary}
    \newtheorem{proposition}[theorem]{Proposition}
    
    \newtheorem{Definition}[theorem]{Definition}
    
    \newtheorem{Remark}[theorem]{Remark}
    \newenvironment{remark}{\begin{Remark}\rm}{\end{Remark}}
    \newtheorem{Example}[theorem]{Example}
    
    \newtheorem{Assumptions}[theorem]{Assumptions}

    \newenvironment{proof}%
    {\rm \trivlist \item[\hskip \labelsep{\bf Proof. }]}%
    {\hspace*{\fill}$\Box$\endtrivlist}
    {\rm \trivlist \item[\hskip \labelsep{\bf Proof}]}%
    {\hspace*{\fill}$\Box$\endtrivlist}

\begin{document}
\title{Asymptotics for Toeplitz determinants: perturbation of symbols with a gap}
\author{Christophe Charlier\footnote{Institut de Recherche en Math\'ematique et Physique,  Universit\'e
catholique de Louvain, Chemin du Cyclotron 2, B-1348
Louvain-La-Neuve, BELGIUM}
 \ and Tom Claeys\footnotemark[\value{footnote}] }

\maketitle

\begin{abstract}
We study the determinants of Toeplitz matrices as the size of the matrices tends to infinity, in the particular case where the symbol has two jump discontinuities and tends to zero on an arc of the unit circle at a sufficiently fast rate. 
We generalize an asymptotic expansion by Widom \cite{Widom}, which was known for symbols supported on an arc.
We highlight applications of our results in the Circular Unitary Ensemble and in the study of Fredholm determinants associated to the sine kernel. 
\end{abstract}

\section{Introduction}

We consider Toeplitz determinants of the form
\begin{equation}\label{Toeplitz}
D_n(f) = \det(f_{j-k})_{j,k=0,\ldots, n-1},\qquad f_{j} = \frac{1}{2\pi} \int_{0}^{2\pi} f(e^{i\theta}) e^{-ij\theta} d\theta,
\end{equation}
where the symbol $f$ is given by
\begin{equation}\label{symbol}
f(e^{i\theta}) = e^{W(e^{i\theta})}\ \times\ 
\begin{cases}
1,&\mbox{ for $-\theta_0\leq \theta\leq \theta_0$,}\\
s,&\mbox{ for $\theta_0<\theta<2\pi-\theta_0$,}
\end{cases}
\end{equation}
and $f_{k}$ is the $k$-th Fourier coefficient of $f$. The symbol depends on the parameters $s\in [0,1]$ and $\theta_0 \in (0,\pi)$, and on a function $W(z)$ which we assume to be analytic in a neighbourhood of the unit circle. In the simplest case where $W(z)=0$, $f$ is piecewise constant with jump discontinuities at $e^{\pm i\theta_0}$.

For $s\in [0,1]$ and $\theta_0\in(0,\pi)$ fixed, the large $n$ asymptotic behavior for the Toeplitz determinants $D_n(f)=D_n(s,\theta_0,W)$ is well understood.
For $s=0$, the symbol is supported on the arc $\gamma= \{e^{i\theta}:-\theta_0\leq \theta\leq \theta_0\}$. If $f$ is positive on $\gamma$ and symmetric, $f(e^{i\theta}) = f(e^{-i\theta})$, we have the following asymptotics due to Widom \cite{Widom},
\begin{multline} \label{Asymptotic_one_arc}
\ln D_n \left(s=0, \theta_{0},W \right) = n^2\ln \sin  \frac{\theta_0}{2}  + n\widetilde{W}_0 - \frac{1}{4} \ln n  \\
+\sum_{k=1}^{\infty}k\widetilde{W}_k\widetilde{W}_{-k}-\frac{1}{4}\ln\cos\frac{\theta_0}{2}+\frac{1}{12}\ln 2 +3\zeta'(-1)+ o(1), \qquad\mbox{ as $n\to\infty$},
\end{multline}
where $\widetilde{W}_k$ is the $k$-th Fourier coefficient of $\widetilde{W}$, defined by
\begin{equation}
\widetilde{W}(e^{i\theta}) = W(e^{2i\arcsin ( \sin \frac{\theta_{0}}{2}\sin \frac{\theta}{2} )}),
\end{equation}
and where $\zeta$ is the Riemann $\zeta$-function.
For $W=0$, it was shown in \cite{Krasovsky, Krasovsky2} that this result holds not only for $0<\theta_0<\pi$ fixed, but also if $\theta_0$ approaches $\pi$ slowly enough such that $n(\pi-\theta_0)$ is large. The error term in (\ref{Asymptotic_one_arc}) then becomes $\bigO(n^{-1}(\pi-\theta_0)^{-1})$.

On the other end of the parameter range for $s$, we have $s=1$: here the symbol is smooth and we have the Szeg\H{o} asymptotics
\begin{equation} \label{Asymptotic_Szego}
\ln D_n \left(s=1, \theta_{0}, W \right) = nW_0+\sum_{k=1}^{\infty}kW_kW_{-k}+ o(1), \qquad \mbox{ as } n \to \infty.
\end{equation}
For $0 < s < 1$, $f$ has two jump discontinuities which are special cases of Fisher-Hartwig (FH) singularities. FH singularities are generally characterized by two parameters: $\alpha$, which describes a root-type singularity, and $\beta$, which describes a jump discontinuity. Our symbol $f$ has Fisher-Hartwig singularities at $e^{\pm i\theta_0}$ with parameters $\alpha=0$ (which means that there are no root-type singularities) and $\beta=\mp \frac{1}{2\pi i} \log s$.
Asymptotics for Toeplitz determinants with Fisher-Hartwig singularities have been studied by many authors \cite{FH, W, Basor, Basor2, BS, Ehr, Deift2, DeiftItsKrasovsky}. Applied to our symbol $f$, the results from, e.g., \cite{Deift2} imply that, for $s\in (0,1)$ fixed, i.e.\ independent of $n$, we have
\begin{multline} \label{Asymptotic_FH}
\ln D_n \left(s,\theta_{0},W \right) = nW_0 + \frac{(\ln s)^2}{2 \pi^{2}}  \ln n +\frac{\ln(2\sin\theta_0)}{2\pi^2} \\
\qquad\qquad +\frac{\ln s}{\pi} \sum_{k=1}^{+\infty}(W_k+W_{-k}) \sin(k\theta_0) 
+ \sum_{k=1}^{+\infty}kW_kW_{-k}\\
+2\ln\left(G\left(1+\frac{\ln s}{2\pi i}\right)G\left(1-\frac{\ln s}{2\pi i}\right)\right) + o(1),
\end{multline}
as $n\to\infty$, where $G$ is Barnes' $G$-function.

If we let $s$ tend to $0$, the jump discontinuities of $f$ turn into endpoints of the support of $f$, and $f$ transforms from a Fisher-Hartwig symbol to an arc-supported symbol on $[-\theta_0,\theta_0]$. Nevertheless, letting $s\to 0$ in \eqref{Asymptotic_FH}, we do not recover the Widom asymptotics \eqref{Asymptotic_one_arc}. This means that a non-trivial critical transition in the asymptotic behavior for $D_n(s,\theta_0,W)$ takes place as $s\to 0$. In this paper, we show that the Widom asymptotics (\ref{Asymptotic_one_arc}) extend to the case where $s\neq 0$ but $s=s(n)\to 0$ sufficiently rapidly as $n\to\infty$. 

\begin{theorem}\label{theorem: extensionWidom}
Let $\theta_0\in (0,\pi)$ and define
\begin{equation}\label{xcintro}
x_{c} = -2\ln \tan \frac{\theta_{0}}{4}.
\end{equation} Let $f$ be of the form (\ref{symbol}) with $W$ analytic in a neighbourhood of the unit circle.
As $n\to\infty$ and simultaneously $s\to 0$ in such a way that $0\leq s\leq e^{-x_cn}$, we have
\begin{equation}\label{extension Widom}
\ln D_n(s,\theta_0,W)=\ln D_n(0,\theta_0,W)+o(1).
\end{equation}
The error term $o(1)$ is uniform for $0\leq s\leq e^{-x_c n}$ and $\epsilon \leq \theta_0 \leq \pi-\epsilon$, $\epsilon>0$, and can be specified as
\begin{equation}\label{error}
o(1)=\bigO(n^{-1/2} e^{x_cn} s).
\end{equation}
In addition, the result extends to the case where
$\theta_0$ approaches $\pi$ at a sufficiently slow rate:  (\ref{extension Widom}) holds for $\epsilon<\theta_0<\pi-\frac{M}{n}$ with $M$ sufficiently large and $0\leq s\leq e^{-x_c n}$, with the error term given by
\begin{equation}\label{error large y}
o(1) = \bigO((\pi-\theta_0)^{1/2}n^{-1/2}e^{nx_c}s).
\end{equation}
\end{theorem}

\begin{figure}[t]
\begin{center}
\includegraphics[width=0.7\textwidth]{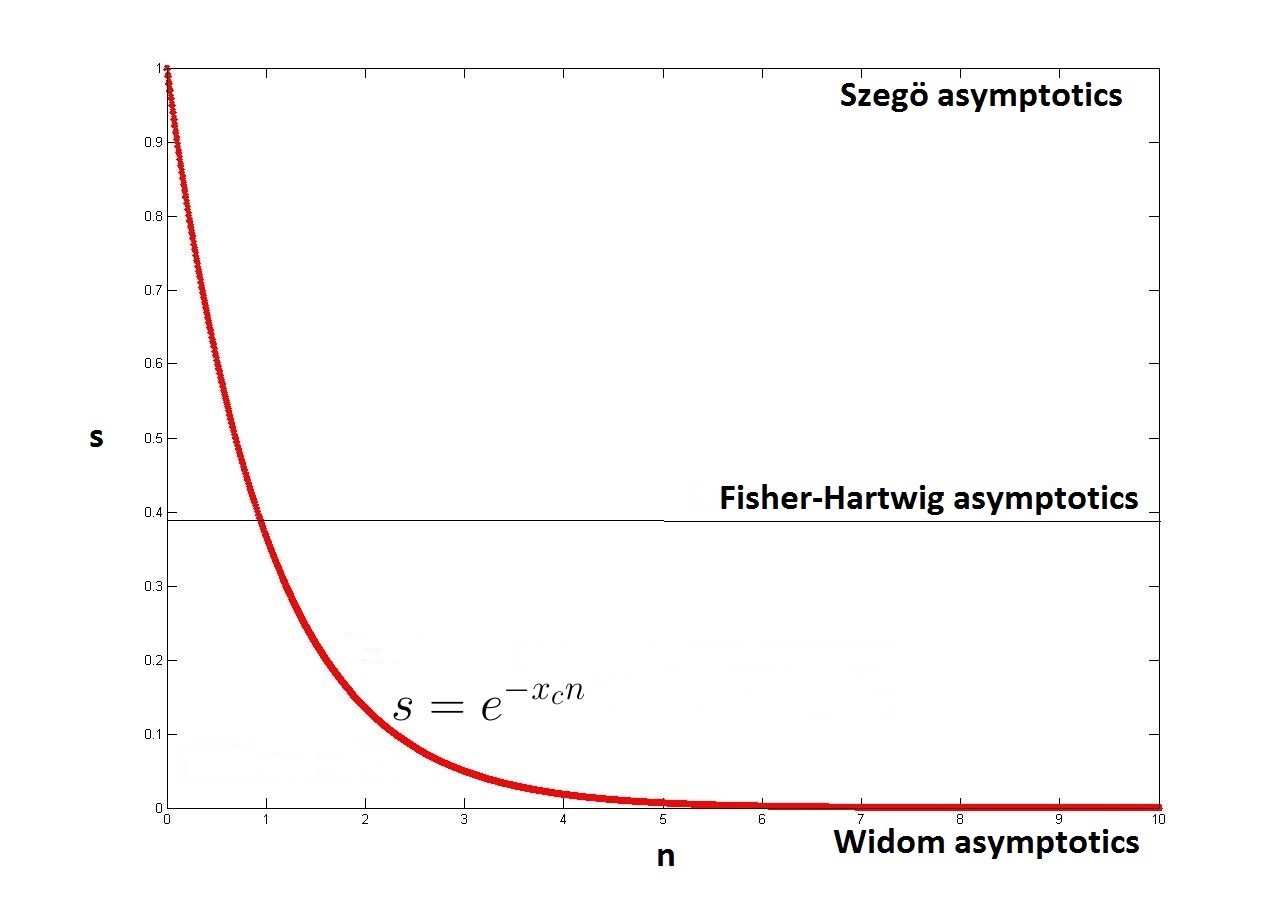}
\end{center}
\caption{\label{fig_s_n} As $n\to\infty$, $\ln D_n(s,\theta_0,W)$ follows Widom asymptotics for $s=0$, Fisher-Hartwig asymptotics for $s\in(0,1)$ fixed, and Szeg\H{o} asymptotics for $s=1$. Theorem \ref{theorem: extensionWidom} implies that the Widom asymptotics remain valid as $n\to\infty$ with $s=s(n)$ below the curve $s=e^{-x_cn}$. As $n\to\infty$ with $s=s(n)$ above the curve, a different type of asymptotic behavior is expected.}
\end{figure}

\begin{remark}\label{remark: elliptic}
We believe the bound $s\leq e^{-x_cn}$ is sharp: as $n\to\infty$, $s\to 0$ with $s>e^{-x_cn}$, the asymptotic behavior for $\ln D_n(s,\theta_0,W)$ is expected to be described in terms of elliptic $\theta$-functions. 
For a heuristic discussion, see Section \ref{section: elliptic}.
\end{remark}
\begin{remark}
If $f$ is positive on $\gamma$ and even in $\theta$, the asymptotics for $\ln D_n(0,\theta_0,W)$ in (\ref{extension Widom}) are given by (\ref{Asymptotic_one_arc}). 
The perturbative result (\ref{extension Widom}) is valid for general analytic $W$, even if the positivity and symmetry conditions needed for the Widom asymptotics do not hold. Note that the error term $o(1)=\bigO((\pi-\theta_0)^{1/2}n^{-1/2}e^{nx_c}s)$ in (\ref{extension Widom}) improves as $\theta_0$ approaches $\pi$. This is reasonable since the perturbation of the arc-supported symbol then takes place on a shrinking arc near $-1$. For our proof, it is however crucial that $n(\pi-\theta_0)$ is sufficiently large.
\end{remark}

\subsection*{Application 1: the Circular Unitary Ensemble}
Consider the Circular Unitary Ensemble (CUE) which is the set of unitary $n\times n$ matrices with the Haar measure.
The eigenvalues $e^{i\theta_1},\ldots, e^{i\theta_n}$ in this ensemble have the joint probability distribution
\begin{equation}\label{jpdf}
\frac{1}{(2\pi)^n n!}\prod_{1\leq j<k\leq n}|e^{i\theta_j}-e^{i\theta_k}|^2 \prod_{j=1}^nd\theta_j,\qquad \theta_1,\ldots, \theta_n\in [0,2\pi).
\end{equation}
Define the random variable $X=X_{\theta_0}$ as the number of eigenvalues of a CUE matrix on the arc $\gamma=\{e^{i\theta}:-\theta_0\leq \theta\leq \theta_0\}$. The average value of $X_{\theta_0}$ is given by
\begin{equation}
\mathbb E_n(X_{\theta_0})=\frac{\theta_0}{\pi}n.
\end{equation}
Define the moment generating function
\begin{equation}
F_n(\lambda)=\mathbb E_n(e^{\lambda X_{\theta_0}}).
\end{equation}
By (\ref{jpdf}), we can write this as
\begin{eqnarray}
F_n(\lambda)&=&\frac{1}{(2\pi)^n n!}\int_{[0,2\pi]^n}\prod_{1\leq j<k\leq n}|e^{i\theta_j}-e^{i\theta_k}|^2 e^{\lambda X_{\theta_0}}\prod_{j=1}^nd\theta_j\\
&=&
\frac{1}{(2\pi)^n n!}\int_{[0,2\pi]^n}\prod_{1\leq j<k\leq n}|e^{i\theta_j}-e^{i\theta_k}|^2 \prod_{j=1}^ne^{\lambda\chi_\gamma(e^{i\theta_j})}d\theta_j,\label{F integral}
\end{eqnarray}
with $\chi_\gamma$ the characteristic function of the arc $\gamma$.
But using the standard multiple integral representation for Toeplitz determinants, we obtain
\begin{equation}
F_n(\lambda)=e^{n\lambda}D_n(s=e^{-\lambda},\theta_0, W=0).
\end{equation}
This is an identity for any $\lambda\geq 0$.
The moment generating function $F_n(\lambda)$ at large positive values of $\lambda$ contains statistical information about large deviations where $X_{\theta_0}$ is much bigger than its average value. As an illustration, a rough estimate of the integral in (\ref{F integral}) gives us the inequality
\begin{equation}\label{est1}
{\rm Prob}_n\left[X_{\theta_0}\geq p\right]\leq e^{-p\lambda}F(\lambda)=e^{(n-p)\lambda}D_n(e^{-\lambda},\theta_0, 0),
\end{equation}
which holds for any $\lambda\geq 0$ and $p\in\{0,1,\ldots, n\}$. Theorem \ref{theorem: extensionWidom} gives asymptotics for $D_n(e^{-\lambda},\theta_0,0)$ for $\lambda\geq nx_c$. For large $n$, the minimum of the right hand side of (\ref{est1}) in the range $\lambda\geq nx_c$ is attained near
\begin{equation}\lambda=nx_c=-2n\ln\tan\frac{\theta_0}{4}.
\end{equation}
Substituting this value for $\lambda$ and the asymptotic behavior for $D_n(e^{-nx_c},\theta_0,0)$, we obtain
\begin{multline}
{\rm Prob}_n\left[X_{\theta_0}\geq p\right]\leq 
\left(\tan\frac{\theta_0}{4}\right)^{-2n^2+2pn}
\left(\sin\frac{\theta_0}{2}\right)^{n^2}
\\
\times\quad n^{-1/4}\left( \cos \frac{\theta_{0}}{2} \right)^{-1/4} 2^{1/12}e^{3\zeta'(-1)}(1+o(1)),\qquad n\to\infty.
\end{multline}
For typical values of $p\sim \frac{\theta_0}{\pi}n$, and even more generally for $p < (1-c)n$ with 
\begin{equation}
c=\frac{\ln\sin\frac{\theta_0}{2}}{2\ln\tan\frac{\theta_0}{4}},
\end{equation} this estimate does not give any useful information since the right hand side is large for large $n$.
We do get non-trivial information for large deviations from the typical value. If we let $p=n-\omega(n)$, we have
\begin{multline}
{\rm Prob_n}\left[X_{\theta_0}\geq p\right]\leq 
\left(\tan\frac{\theta_0}{4}\right)^{-2n\omega(n)}
\left(\sin\frac{\theta_0}{2}\right)^{n^2}
\\
\times \quad n^{-1/4}\left( \cos \frac{\theta_{0}}{2} \right)^{-1/4}2^{1/12}e^{3\zeta'(-1)}(1+o(1)),
\end{multline}
as $n\to\infty$,
and the right hand side decays whenever $\omega(n)<cn$.
If $\omega(n)=o(n)$ as $n\to\infty$, we observe that the probability to have at least $n-\omega(n)$ eigenvalues on $\gamma$ can be estimated by
\begin{equation}
{\rm Prob_n}\left[X_{\theta_0}\geq n-\omega(n)\right]=\bigO(n^{-1/4}e^{n^2\phi(\theta_0)-2n\omega(n)\ln\tan\frac{\theta_0}{4}}), \qquad \mbox{ as $n\to\infty$,}\label{LDP}
\end{equation}
where the rate function $\phi$ is given by $\phi(\theta_0)=\ln\sin\frac{\theta_0}{2}$.
Note that this is not a large deviation principle since we only have an asymptotic upper bound for the left hand side of (\ref{LDP}). A more detailed analysis of the moment generating function $F(\lambda)$ may lead to better estimates, but it is not our aim to proceed in this direction.

\subsection*{Application 2: sine kernel determinants}\label{section: Fredholm}
In the double scaling limit where $n\to\infty$ and $\theta_0$ approaches $\pi$ at rate $\bigO(1/n)$, the Toeplitz determinant $D_n(s,\theta_0,0)$ converges to the Fredholm determinant of the operator $1-(1-s)K_y$, where $K_y$ is the integral operator acting on $(-y,y)$ with kernel $\frac{\sin\pi(x-t)}{\pi(x-t)}$: we have \cite[formulas (236)-(241)]{DIK-hist}
\begin{equation}\label{Toeplitz Fredholm}
\lim_{n\to\infty}D_n(s,\theta_0=\pi (1-\frac{2y}{n}),0)=\det(1-(1-s)K_y),
\end{equation}
where the operator $K_y$ is defined by
\begin{equation}
(K_{y} g)(x)=\int_{-y}^{y}\frac{\sin \pi(x-t)}{\pi(x-t)}g(t)dt.
\end{equation}
Asymptotics for the right hand side of (\ref{Toeplitz Fredholm}) have been studied by Dyson \cite{Dyson}, see also \cite{DIK-hist} for a historical review; the double scaling limit where $y\to\infty$ and simultaneously $s\to 0$ appears to be subtle.
As $y\to\infty$, $s\to 0$ in such a way that $s\leq e^{-2\pi y}$, with
\begin{equation}\label{constant}
c=\frac{1}{12}\ln 2+3\zeta'(-1),
\end{equation} 
we have
\begin{equation}\label{large gap}
\ln \det(1-(1-s)K_y)=-\frac{\pi^{2}y^2}{2}-\frac{1}{4}\ln \pi y+c+o(1).
\end{equation}
Except for $s=0$, (\ref{large gap}) had not been proved rigorously until the recent paper \cite{BDIK}, where double scaling asymptotics were obtained both for $s\leq e^{-2\pi y}$ and for $s> e^{-2\pi y}$. The value of the constant (\ref{constant}) was also proved in \cite{BDIK}.

As a consequence of our Theorem \ref{theorem: extensionWidom}, we re-derive (\ref{large gap}) and (\ref{constant}).
Indeed,
by (\ref{Asymptotic_one_arc}), (\ref{extension Widom}), and (\ref{error large y}), we have 
\begin{equation}\label{large y}
\ln D_n(s,\pi(1-\frac{2y}{n}),0)=n^2\ln\cos\frac{\pi y}{n}-\frac{1}{4}\ln n-\frac{1}{4}\ln\sin\frac{\pi y}{n}+c+\bigO(\frac{y^{1/2}}{n}e^{nx_{c}}s)+ \bigO(y^{-1}),
\end{equation}
with the error term uniform for $n,y$ large, $y^{1/2}/n$ small, and 
\begin{equation}\label{condition sy}s\leq e^{-x_cn}=e^{2n\ln\tan\frac{\theta_0}{4}}=e^{2n\ln\tan\left(\frac{\pi}{4}-\frac{\pi y}{2n}\right)}.\end{equation}
The right hand side of (\ref{condition sy}) can be written as $e^{-2\pi y}\left(1+\bigO(y^2 n^{-1})\right)$ for $n,y$ large and $y^{2}/n$ small. \\
Now, let us choose an arbitrary small function $\epsilon (y) > 0$ and assume $s \leq (1-\epsilon(y))e^{-2\pi y}$. Then for $n$ sufficiently large, i.e.\ for $n>>y^2/\epsilon(y)$, (\ref{condition sy}) holds, and letting $n\to\infty$ in (\ref{large y}), we obtain  
\begin{equation}
\lim_{n\to\infty}\ln D_n(s,\pi(1-\frac{2y}{n}),0)=-\frac{\pi^{2}y^2}{2}-\frac{1}{4}\ln \pi y+c+\bigO(y^{-1}),
\end{equation}
for large $y$.
We thus re-derive (\ref{large gap}) in the double scaling limit where $y\to\infty, s\to 0$ in such a way that $s\leq (1-\epsilon(y))e^{-2\pi y}$, and we confirm the value of the constant (\ref{constant}), using different methods than in \cite{BDIK}.

In \cite{BDIK}, the double scaling asymptotics $y\to\infty, s\to 0$ such that $s\geq e^{-2\pi y}$ have also been considered, and asymptotics for $\det(1-(1-s)K_y)$ in terms of elliptic $\theta$-functions were obtained. This is consistent with Remark \ref{remark: elliptic} and the discussion in Section \ref{section: elliptic}.

\subsection*{Outline}
In Section \ref{section: diff id}, we will derive identities which express $\frac{d}{ds}\ln D_n(s,\theta_0,W)$ in terms of orthogonal polynomials on the unit circle.
In Section \ref{section: eq}, we will study an equilibrium problem which is crucial to obtain asymptotics for the orthogonal polynomials. In Section \ref{section: RH}, we will obtain asymptotics for the orthogonal polynomials as the degree tends to infinity, via the Riemann-Hilbert (RH) method. An important novel feature in the RH analysis is the construction of a {\em modified Bessel parametrix} compared to Bessel parametrices that appeared in the literature before. This modification is needed because $s\neq 0$. The asymptotics for the orthogonal polynomials will lead us towards large $n$ asymptotics for $\frac{d}{ds}\ln D_n(s,\theta_0,W)$. Integrating the differential identity will complete the proof of Theorem \ref{theorem: extensionWidom} 
in Section \ref{section: int} for $\theta_{0}$ fixed. In addition, we explain how the proof can be extended to the case where $\theta_{0} \to \pi$ at a sufficiently slow rate.

\section{Differential identities for Toeplitz determinants}\label{section: diff id}

In this section, we derive an identity for $\frac{d}{ds}\ln D_n(s,\theta_0,W)$ in terms of orthogonal polynomials on the unit circle.
Define $\phi_{j}, \widehat{\phi_j}$, $j=0,1,2,\ldots$ as the family of orthogonal polynomials on the unit circle characterized by
\begin{equation} \label{Orthogonality_condition}
\frac{1}{2\pi} \int_{0}^{2\pi} \phi_{k}(e^{i\theta}) \widehat{\phi_{m}}(e^{-i\theta})f(e^{i\theta})d\theta = \delta_{km}, \hspace{1cm} \forall k,m \in \mathbb{N},
\end{equation}
where the degree of $\phi_{j}, \widehat{\phi_j}$ is $j$, and their leading coefficients $\chi_j$ are equal and positive.

For general weight functions $f(e^{i\theta})$, it is a standard fact that $\phi_n,\widehat\phi_n$ exist and are unique if $D_n(f), D_{n+1}(f)\neq 0$ and that they are given by the determinantal formulas \cite{Simon}
\begin{equation} \label{explicit OP}
\phi_{n}(z) = \frac{1}{\sqrt{D_{n}D_{n+1}}} \left| \begin{array}{c c c c c}
f_0 & f_{-1} & f_{-2} & \cdots & f_{-n} \\
f_{1} & f_{0} & f_{-1} & \cdots & f_{-n+1} \\
\vdots & \vdots & \vdots & \ddots & \vdots \\
f_{n-1} & f_{n-2} & f_{n-3} & \cdots & f_{-1} \\
1 & z & z^2 & \cdots & z^n
\end{array} \right|,
\end{equation}
and
\begin{equation} \label{explicit OP2}
\widehat\phi_{n}(z) = \frac{1}{\sqrt{D_{n}D_{n+1}}} \left| \begin{array}{c c c c c c}
f_0 & f_{-1} & f_{-2} & \cdots & f_{-n+1} &1\\
f_{1} & f_{0} & f_{-1} & \cdots & f_{-n+2}& z\\
\vdots & \vdots & \vdots & \ddots & \vdots& \vdots\\
f_{n} & f_{n-1} & f_{n-2} & \cdots & f_{1}& z^{n}
\end{array} \right|,
\end{equation}
where $D_n=D_n(f)$ is the Toeplitz determinant defined in (\ref{Toeplitz}), and where $f_k$ is, as before, the $k$-th Fourier coefficient of $f$.
If $f(e^{i\theta})$ is  positive, $D_n(f)>0$ for all $n\in\mathbb N$ and the OPs are well-defined. 
For complex $f$, existence and uniqueness of the orthogonal polynomials is not guaranteed for any $n$.
For $f$ given in (\ref{symbol}), it will follow from our analysis that they are well-defined for $n$ sufficiently large and $0\leq s\leq e^{-x_cn}$.
If $W=0$, the symbol $f$ is even in $\theta$ and positive, which implies that $\widehat\phi_k=\phi_{k} = \overline{\phi_{k}}$.

From (\ref{explicit OP}), we see that the leading coefficient $\chi_n$ of $\phi_n$ is equal to
\begin{equation*}
\chi_n = \chi_n(s,\theta_0,W) = \sqrt{\frac{D_{n}(s,\theta_0,W)}{D_{n+1}(s,\theta_0,W)}}.
\end{equation*}
We can thus express $D_n(s,\theta_0,W)$ in terms of the leading coefficients of the OPs if $D_0, D_1, \ldots, D_n\neq 0$:
\begin{equation}\label{Toeplit_express_by_chi_n}
D_{n}(s,\theta_0,W) = \prod_{j=0}^{n-1} \chi_{j}^{-2}.
\end{equation}

To obtain an asymptotic formula for $D_{n}(s,\theta_0,W)$ from this formula, we would need information about all OPs of degree $0$ up to $n-1$. To circumvent this problem, we derive differential identities for $\ln D_{n}(s,\theta_0,W)$.

\subsection{Differential identity for general $W$}

For general sufficiently smooth symbols $f(z;s)$ depending on a parameter $s$, there exists an identity for the logarithmic derivative of $D_n(f)$ with respect to $s$: if $D_{n-1}, D_n, D_{n+1}\neq 0$, we have \cite[Proposition 3.3]{DeiftItsKrasovsky}
\begin{equation*}
\partial_{s} \ln D_{n} (s,\theta_{0},W) = \frac{1}{2\pi i} \int_{S_{1}} z^{-n} \left[ Y^{-1}(z) Y^{\prime}(z) \right]_{21} \partial_s f(z,s) dz,
\end{equation*}
where $S_1$ is the unit circle in the complex plane, $\mbox{ }^{\prime} = \frac{d}{dz}$, $ \partial_s = \frac{d}{ds}$, and
\begin{equation}\label{sol_Y}
Y(z) = \begin{pmatrix} \chi_{n}^{-1} \phi_{n}(z) & \displaystyle \chi_{n}^{-1} \int_{S_{1}} \frac{\phi_{n}(w)}{w-z} \frac{f(w)}{2\pi i w^{n}} dw \\
\displaystyle -\chi_{n-1} z^{n-1} \widehat{\phi}_{n-1}(z^{-1}) & \displaystyle -\chi_{n-1} \int_{S_{1}} \frac{\widehat{\phi}_{n-1}(w^{-1})}{w-z} \frac{f(w)}{2\pi i w} dw
\end{pmatrix}.
\end{equation}
For our symbol given by (\ref{symbol}), this formula reduces to
\begin{equation}\label{diffid1general}
\partial_{s} \ln D_{n} (s,\theta_{0},W) = \frac{1}{2\pi i} \int_{\gamma^{c}} z^{-n} \left[ Y^{-1}(z)Y^{\prime}(z) \right]_{21} e^{W(z)}dz,
\end{equation}
where $\gamma^{c} = S_{1}\setminus \gamma$.

\subsection{Differential identity for $W = 0$}

If $W=0$, the integral at the right hand side of (\ref{diffid1general}) can be simplified. Although (\ref{diffid1general}) will be sufficient for the proof of Theorem \ref{theorem: extensionWidom}, we will show here how to simplify the differential identity if $W=0$. This will allow us to give a more elegant proof of Theorem \ref{theorem: extensionWidom} in this case.

\begin{proposition}
\label{proposition: diff id}
Let $D_n(s,\theta_0,0)$ be the Toeplitz determinant with symbol (\ref{symbol}) in the case where $W(e^{i\theta})= 0$.
We have the following differential identity for $\ln D_{n}$:
\begin{equation}\label{Differential_identity_for_D_n}
\partial_{s} \ln D_n(s,\theta_{0},0) = -2n \frac{\partial_{s} \chi_n}{\chi_n} + \frac{2(1-s)}{\pi} \mbox{\rm Im}\left( \overline{\phi_{n}(e^{i\theta_{0}})} \partial_{s}\phi_{n}(e^{i\theta_{0}}) \right).
\end{equation}
\end{proposition}
\begin{proof}
Taking the logarithm of both sides in (\ref{Toeplit_express_by_chi_n}), and then differentiating with respect to $s$, we get
\begin{equation*}
\partial_{s} \ln D_n(s,\theta_{0},0) = -2 \sum_{j=0}^{n-1} \frac{\partial_{s}\chi_{j}}{\chi_{j}}.
\end{equation*}
On the other hand, by \eqref{Orthogonality_condition},
\begin{equation*}
\frac{1}{2\pi} \int_{0}^{2\pi} \partial_{s} \left( \phi_{j}(e^{i\theta}) \overline{\phi_{j}(e^{i\theta})} \right) f(e^{i\theta})d\theta = 2\frac{\partial_{s}\chi_{j}}{\chi_{j}},
\end{equation*}
and this gives
\begin{equation}\label{diffid1}
\partial_{s} \ln D_n(s,\theta_{0},0) = - \frac{1}{2\pi} \int_{0}^{2\pi} \partial_{s} \left[ \sum_{j=0}^{n-1} \phi_{j}(e^{i\theta}) \overline{\phi_{j}(e^{i\theta})} \right] f(e^{i\theta})d\theta.
\end{equation}
Here we can use the Christoffel-Darboux formula (see e.g.\ \cite{Simon, Deift2} for a proof of it):
\begin{equation*}
\sum_{j=0}^{n-1} \phi_{j}(z) \overline{\phi_{j}}(z^{-1}) = -n\phi_{n}(z) \overline{\phi_{n}}(z^{-1}) + z\left( \overline{\phi_{n}} (z^{-1}) \phi_{n}^{\prime}(z) - \left( \overline{\phi_{n}}(z^{-1}) \right)^{\prime} \phi_{n}(z) \right).
\end{equation*}
Substituting this into (\ref{diffid1}), we obtain
\begin{align} 
&\partial_{s}\ln D_{n}(s,\theta_{0},0)  =  2n \frac{\partial_{s} \chi_{n}}{\chi_{n}}\nonumber \\ &- \frac{1}{2\pi} \int_{0}^{2\pi} \partial_{s} \left[ e^{i\theta}\left( \overline{\phi_{n}} (e^{-i\theta}) \phi_{n}^{\prime}(e^{i\theta}) - \left. \left( \overline{\phi_{n}}(z^{-1}) \right)^{\prime}\right|_{z=e^{i\theta}} \phi_{n}(e^{i\theta}) \right) \right] f(e^{i\theta})d\theta \nonumber\\
& \qquad =  2n \frac{\partial_{s} \chi_{n}}{\chi_{n}} + I_{1} + I_{2} + I_{3} + I_{4},
\label{diffid2}\end{align}
where
\begin{align}
&I_{1}  =  - \frac{1}{2\pi} \int_{0}^{2\pi} e^{i\theta} \partial_{s} \left( \overline{\phi_{n}} (e^{-i\theta}) \right) \phi_{n}^{\prime} (e^{i\theta}) f(e^{i\theta})d\theta,\\
&I_{2} = -\frac{1}{2\pi} \int_{0}^{2\pi} e^{i\theta} \overline{\phi_{n}}(e^{-i\theta})  \partial_{s} \left( \phi_{n}^{\prime}(e^{i\theta}) \right) f(e^{i\theta})d\theta,\\
&I_{3} = \frac{1}{2\pi} \int_{0}^{2\pi} e^{i\theta} \partial_{s} \left( \left. \left( \overline{\phi_{n}}(z^{-1}) \right)^{\prime} \right|_{z=e^{i\theta}} \right) \phi_{n}(e^{i\theta}) f(e^{i\theta})d\theta,\\
&I_4=\frac{1}{2\pi} \int_{0}^{2\pi} e^{i\theta} \left. \left( \overline{\phi_{n}}(z^{-1}) \right)^{\prime} \right|_{z=e^{i\theta}} \partial_{s} \phi_{n}(e^{i\theta}) f(e^{i\theta})d\theta.
\end{align}
From the orthogonality relation (\ref{Orthogonality_condition}), we easily get
\begin{equation}
I_2=I_3=-n \frac{\partial_s\chi_{n}}{\chi_{n}}.\label{I23}
\end{equation}
The computation of $I_1$ and $I_4$ is slightly more involved. Using (\ref{symbol}), we have
\begin{multline} 
I_{1}  =  - \frac{1}{2\pi i} \int_{0}^{2\pi} \partial_{s} \left( \overline{\phi_{n}} (e^{-i\theta}) \right) \frac{d}{d\theta} \left( \phi_{n} (e^{i\theta}) \right) d\theta \\
+ \frac{1-s}{2\pi i} \int_{\theta_0}^{2\pi-\theta_0} \partial_{s} \left( \overline{\phi_{n}} (e^{-i\theta}) \right) \frac{d}{d\theta} \left( \phi_{n} (e^{i\theta}) \right) d\theta.
\end{multline}
Integrating by parts and then using orthogonality, we obtain
\begin{eqnarray}
I_1&=&\frac{1-s}{2\pi i} \left[ \partial_{s} \left( \overline{\phi_{n}}(e^{-i\theta}) \right) \phi_{n}(e^{i\theta}) \right]_{\theta_{0}}^{2\pi -\theta_{0}} + \frac{1}{2\pi i} \int_{0}^{2\pi} \phi_{n}(e^{i\theta}) \partial_{s}\frac{d}{d\theta} \left( \overline{\phi_{n}} (e^{-i\theta}) \right) f(e^{i\theta}) d\theta\nonumber \\&=& \frac{1-s}{2\pi i} \left[ \partial_{s} \left( \overline{\phi_{n}}(e^{-i\theta}) \right) \phi_{n}(e^{i\theta}) \right]_{\theta_{0}}^{2\pi -\theta_{0}} - n \frac{\partial_s \chi_{n}}{\chi_{n}}.\label{I1}
\end{eqnarray}
In the same way, we show that
\begin{eqnarray}
\begin{array}{r c l}
\displaystyle I_{4} & = & \displaystyle  \frac{1}{2\pi} \int_{0}^{2\pi} e^{i\theta} \left. \left( \overline{\phi_{n}}(z^{-1}) \right)^{\prime} \right|_{z=e^{i\theta}} \partial_{s} \phi_{n}(e^{i\theta}) f(e^{i\theta})d\theta \\
& = &  \displaystyle -\frac{1-s}{2\pi i} \left[ \overline{\phi_{n}}(e^{-i\theta}) \partial_{s} \phi_{n}(e^{i\theta}) \right]_{\theta_{0}}^{2\pi -\theta_{0}} - n \frac{\partial_s \chi_{n}}{\chi_{n}}.\label{I4}
\end{array}
\end{eqnarray}
Summing up (\ref{I23}), (\ref{I1}), and (\ref{I4}), and using the fact that $\overline{\phi_{n}} =  \phi_{n}$ if $W=0$, we get the result.
\end{proof}

As a consequence of Proposition \ref{proposition: diff id} and (\ref{sol_Y}), we can express the right hand side of (\ref{Differential_identity_for_D_n}) in terms of $Y=Y^{(n)}$ given by (\ref{sol_Y}):

\begin{corollary}\label{cor}We have
\begin{equation} \label{D_n_in_term_of_Y}
\partial_{s} \ln D_{n}(s,\theta_{0},0) = n \partial_{s}\ln Y_{12}(0) + \frac{2(1-s)}{\pi} \Im \left( \frac{\overline{Y_{11}(e^{i\theta_{0}})}}{\sqrt{Y_{12}(0)}} \partial_{s} \left( \frac{Y_{11}(e^{i\theta_{0}})}{\sqrt{Y_{12}(0)}} \right) \right).
\end{equation}
\end{corollary}
\begin{proof}
By \eqref{sol_Y}, we have
\begin{equation*}
Y_{12}(0) = \chi_{n}^{-2},\qquad Y_{11}(e^{i\theta_{0}}) = \chi_{n}^{-1}\phi_{n}(e^{i\theta_{0}}).
\end{equation*}
Substituting this into \eqref{Differential_identity_for_D_n}, we get \eqref{D_n_in_term_of_Y}.
\end{proof}

Integrating both sides from $0$ to $s_0$ in (\ref{diffid1general}) or (\ref{Differential_identity_for_D_n}), we obtain
\begin{equation}\label{integration_diff_identity}
\ln D_n(s_0,\theta_0,W) - \ln D_n(0,\theta_0,W) = \int_{0}^{s_0} \left[ \partial_s \ln D_n (s,\theta_0,W)\right]ds.
\end{equation}

In order to obtain asymptotics for $\ln D_n(s_0,\theta_0,W)$, we need large $n$ asymptotics for the right-hand side of (\ref{diffid1general}) and \eqref{Differential_identity_for_D_n} uniformly for $0 \leq s \leq s_0$.

To obtain large $n$ asymptotics for $Y$, we will use the following RH characterization \cite{FIK}:
$Y(z) = Y^{(n)}(z)$ is the unique $2 \times 2$ matrix-valued function  which satisfies the following properties. 
\subsubsection*{RH problem for $Y$}
\begin{itemize}
\item[(a)] $Y : \mathbb{C}\setminus S_{1} \to \mathbb{C}^{2\times 2}$ is analytic.
\item[(b)] $Y$ has the following jumps:
\begin{equation*}
Y_{+}(z) = Y_{-}(z) \begin{pmatrix}
1 & z^{-n}f(z) \\ 0 & 1
\end{pmatrix}, \hspace{0.5cm} \mbox{ for } z \in S_{1}\setminus \left\{z_{0}=e^{i\theta_0} ,\overline{z_{0}}=e^{-i\theta_0}\right\},
\end{equation*}
where $Y_+(z)$ ($Y_-(z)$) denotes the limit as $z$ is approached from the inside (outside) of the unit circle.
\item[(c)] $Y(z) = \left(I + \bigO(z^{-1})\right)  \begin{pmatrix}
z^{n} & 0 \\ 0 & z^{-n}
\end{pmatrix}$ as $z \to \infty$.
\item[(d)] As $z$ tends to $z_0$ or $z$ tends to $\overline{z_0}$, $Y$ behaves as
\begin{equation*}
\begin{array}{l l}
Y(z) = \begin{pmatrix}
\bigO(1) & \bigO(\ln |z-z_0|) \\ \bigO(1) & \bigO(\ln |z-z_0|)

\end{pmatrix} &  \mbox{ as } z \to z_0, \\ [0.8cm]
Y(z) = \begin{pmatrix}
\bigO(1) & \bigO(\ln |z-\overline{z_0}|) \\ \bigO(1) & \bigO(\ln |z-\overline{z_0}|)
\end{pmatrix} &  \mbox{ as } z \to \overline{z_0}.
\end{array}
\end{equation*}
\end{itemize}

\section{Equilibrium measures}\label{section: eq}

It will be convenient to introduce a new parameter $x \in \mathbb{R}^{+} \cup \{+\infty\}$ defined by
\begin{equation}\label{sx}
s = e^{-xn},
\end{equation}
such that $f$ can be written as $f(e^{i\theta}) = e^{W(e^{i\theta})}e^{-nV(e^{i\theta})}$, where 
\begin{equation}\label{V}
V(e^{i\theta}) = \begin{cases}
0, & \mbox{ for } e^{i\theta} \in \gamma, \\
x, & \mbox{ for } e^{i\theta} \in S_{1} \setminus \gamma.
\end{cases}
\end{equation}
An important ingredient for the large $n$ analysis of the orthogonal polynomials and the RH problem for $Y$ is the following minimization problem:
find the measure $\mu=\mu^{(x)}$ which minimizes 
\begin{equation} \label{inf_prob}
 \iint \log |z-s|^{-1} d\mu(z)d\mu(s) + \int V(z)d\mu(z)
\end{equation}
among all Borel probability measures on the unit circle $S_{1}$.


This measure, which is unique and absolutely continuous with respect to the Lebesgue measure, is called the equilibrium measure, is denoted $d\mu^{(x)}=u^{(x)}(e^{i\theta})d\theta$, and its support is denoted $J^{(x)}$.
The equilibrium measure $\mu=\mu^{(x)}$ and its support $J=J^{(x)}$ are uniquely determined by the following Euler-Lagrange variational conditions \cite{BDJ, SaTo}:
there exists a real constant $\ell=\ell^{(x)}$ such that
\begin{align}
&\label{EL_1}
2 \int_{-\pi}^{\pi} \log | z-e^{i\theta} | d\mu(e^{i\theta}) - V(z) + \ell = 0, &\mbox{ for } z \in J,\\&\label{EL_2}
2 \int_{-\pi}^{\pi} \log | z-e^{i\theta} | d\mu(e^{i\theta}) - V(z) + \ell \leq 0, &\mbox{ for } z \in S_{1}\setminus J.
\end{align}
The support and density of the equilibrium measure can be computed explicitly.

\begin{proposition}\label{thm200}
\begin{itemize}
\item[(a)] For $x = +\infty$, the support of the equilibrium measure and its density are given by
\begin{equation}\label{uinfty}
J=J^{(\infty)}=\gamma,\qquad u(e^{i\theta})=u^{(\infty)}(e^{i\theta}) = \frac{1}{2\pi} \sqrt{\frac{\cos \theta + 1}{\cos \theta - \cos \theta_{0}}}.
\end{equation}
The constant $\ell=\ell^{(\infty)}$ in the variational conditions (\ref{EL_1})-(\ref{EL_2}) is given by
\begin{equation}\label{linfy}
\ell^{(\infty)} = - 2 \ln \sin \frac{\theta_{0}}{2},
\end{equation}
and the variational inequality (\ref{EL_2}) is strict for $z\in S_1\setminus{\gamma}$.
\item[(b)]Let $x_c$ be given by (\ref{xcintro}). For $x \geq x_{c}$, we have the same result as for $x=+\infty$:
\begin{equation}\label{ux1}
J = \gamma \hspace{0.1cm},\qquad u(e^{i\theta}) = u^{(\infty)}(e^{i\theta}),\qquad \ell = \ell^{(\infty)}.
\end{equation}
Moreover, the variational inequality (\ref{EL_2}) is strict for $z\in S_1\setminus{\gamma}$ if $x\neq x_c$; if $x=x_c$, it is strict for $z\in S_1\setminus\left(\gamma \cup\{-1\}\right)$, and there is equality for $z=-1$.
\item[(c)] For $0 < x < x_{c}$, the support of $\mu$ consists of two disjoint arcs: we have
\begin{equation}\label{ux2}
J = \{ e^{i\theta} : \theta \in [-\theta_{0},\theta_{0}] \cup [\pi-\theta_{1},\pi + \theta_{1}] \} ,
\end{equation}
where $\theta_1=\theta_1(x)\in (0,\pi-\theta_0)$ is the unique solution of 
\begin{equation}\label{bijection_theta1_x}
2\int_{[-\theta_{0},\theta_{0}] \cup [\pi-\theta_{1},\pi + \theta_{1}]} \log \left| \frac{1+e^{i\theta}}{1-e^{i\theta}} \right| u(e^{i\theta}) d\theta = x,
\end{equation}
and the density is given by 
\begin{equation}
u(e^{i\theta}) = \frac{1}{2\pi} \sqrt{\frac{\cos \theta + \cos \theta_{1}}{\cos \theta - \cos \theta_{0}}}.
\end{equation}
The constant $\ell$ is given by
\begin{equation}\label{lx}
\ell = - \int_{0}^{1} \frac{1}{s} \left( 1 - \sqrt{\frac{s^2 + 2\cos (\theta_{1}) s +1}{s^2 - 2\cos (\theta_{0}) s +1}} \right) ds,
\end{equation}
and the variational inequality (\ref{EL_2}) is strict for $z\in S_1\setminus{J}$.
\end{itemize}
\end{proposition}
\begin{remark}
Although part (c) of the proposition is not needed for the proof of Theorem \ref{theorem: extensionWidom}, we present it here for completeness and to support the heuristic arguments in Section \ref{section: elliptic}, where we will discuss asymptotics for $D_n$ if $x< x_c$.
\end{remark}
\begin{proof}
Note first that the equation (\ref{bijection_theta1_x}) has indeed a unique solution $\theta_1$ for $x<x_c$, since the function
\begin{equation*}
\theta_{1} \mapsto 2\int_{[-\theta_{0},\theta_{0}] \cup [\pi-\theta_{1},\pi + \theta_{1}]} \log \left| \frac{1+e^{i\theta}}{1-e^{i\theta}} \right| u(e^{i\theta}) d\theta
\end{equation*} decreases as a function of $\theta_{1} \in (0,\pi-\theta_{0})$, and is bijective from $(0,\pi-\theta_{0})$ to $(0,x_c)$. 

\medskip

The equilibrium measure $\mu$ is uniquely characterized by the conditions \eqref{EL_1}-\eqref{EL_2}, which means that it is sufficient for us to show that the measure $\mu$ defined in cases (a), (b), and (c) satisfy these variational conditions. 
To do so, consider the function
\begin{equation}
f(z) =2\int \log | z-e^{i\theta}| d\mu^{(x)}(e^{i\theta}).
\end{equation}
If we define $\theta_1>0$ as the unique solution of (\ref{bijection_theta1_x}) if $x<x_c$, and if we let $\theta_1=0$ for $x\geq x_c$, we have
\begin{equation}
f(e^{i\alpha})=\frac{1}{\pi}\int_{[-\theta_0,\theta_0]\cup[\pi-\theta_1,\pi+\theta_1]} \log | e^{i\alpha}-e^{i\theta}|\sqrt{\frac{\cos \theta + \cos \theta_{1}}{\cos \theta - \cos \theta_{0}}}d\theta.
\end{equation}
The derivative $\frac{d}{d\alpha} f(e^{i\alpha})$ can be written as a contour integral 
\begin{equation}
\frac{d}{d\alpha} f(e^{i\alpha})=-\frac{1}{2\pi i}\int_\Sigma \frac{1}{\xi}\frac{\xi+e^{i\alpha}}{\xi-e^{i\alpha}}\left(\frac{(\xi-z_1)(\xi-\overline{z_1})}{(\xi-z_0)(\xi-\overline{z_0})}\right)^{1/2}d\xi, \qquad z_{1} = e^{i\theta_{1}},
\end{equation}
where the square root is analytic off $J$ and tends to $1$ as $\xi\to\infty$, and
where the contour $\Sigma$ consists of one (if $x\geq x_c$) or two (if $x<x_c$) counterclockwise oriented circles around the arc(s) of $J$. If $e^{i\alpha}\in S_1\setminus J$, the contour has to be chosen sufficiently small such that $e^{i\alpha}$ lies in the exterior of $\Sigma$.

If $e^{i\alpha} \in J$, a residue calculation shows that $\frac{d}{d\alpha} f(e^{i\alpha})=0$.
If $e^{i\alpha} \in S_{1}\setminus J$ on the other hand, we have
\begin{equation} \label{derivative_of_EL}
\frac{d}{d\alpha} f(e^{i\alpha})  
 = \begin{cases}
 \sqrt{\frac{\cos \theta_{1} + \cos \alpha}{ \cos \theta_{0} - \cos \alpha}}, & \mbox{ if } \theta_{0} < \alpha < \pi-\theta_{1}, \\[0.4cm]
 -\sqrt{\frac{\cos \theta_{1} + \cos \alpha}{ \cos \theta_{0} - \cos \alpha}}, & \mbox{ if } \pi+\theta_{1} < \alpha < 2\pi-\theta_{0}.
\end{cases} 
\end{equation}
It follows that $f(e^{i \alpha})$ is constant on $[-\theta_0,\theta_0]\cup[\pi-\theta_1,\pi+\theta_1]$, and that it achieves its maximum on $[\pi-\theta_1,\pi+\theta_1]$ (i.e.\ at $\pi$ if $x\geq x_c$). 
We have
\begin{align}
&\label{f1}f(e^{i\alpha})=f(1),&\mbox{ for $-\theta_0\leq \alpha\leq \theta_0$},\\
&f(e^{i\alpha})=f(-1),&\mbox{ for $\pi-\theta_1\leq \alpha\leq \pi+\theta_1$},\\
&\label{f3}f(e^{i\alpha})< f(-1),&\mbox{ for $\alpha\notin [\pi-\theta_1,\pi+\theta_1]$}.
\end{align}
If we show that 
\begin{align}
&\label{toshowf1}f(1)=-\ell, \ f(-1)<f(1)+x, & \mbox{ for $x>x_c$},\\
&\label{toshowf2}f(1)=-\ell, \ f(-1)=f(1)+x, & \mbox{ for $x\leq x_c$},
\end{align}
then (\ref{f1})-(\ref{f3}) imply the Euler-Lagrange conditions (\ref{EL_1})-(\ref{EL_2}).

For the case $x\geq x_c$, using  (\ref{derivative_of_EL}), we obtain after a straightforward calculation, \begin{equation*}
f(-1) - f(1) =f(e^{i\pi})-f(e^{i\theta_0})= \int_{\theta_0}^{\pi} \sqrt{\frac{1+\cos \theta}{\cos \theta_{0} - \cos \theta}}d\theta= -2 \ln \tan \frac{\theta_{0}}{4} = x_c \leq x,
\end{equation*}
which proves the inequality in (\ref{toshowf1}), and the fact that there is equality at $-1$ if $x=x_c$.
Moreover, using an other residue calculation, we get
\begin{equation*}
f(1) = \int_{0}^{1} f^{\prime}(s) ds = \int_{0}^{1} \frac{1}{s}\left( 1 - \frac{1+s}{\sqrt{s^{2}-2\cos (\theta_{0})s + 1}} \right)ds = 2\ln \sin \frac{\theta_{0}}{2} = -\ell^{(\infty)}.
\end{equation*}
This proves the equality in (\ref{toshowf1}).
For $0 < x < x_{c}$, by
(\ref{derivative_of_EL}), 
\begin{equation*}
f(-1) - f(1) =f(e^{i(\pi-\theta_1)})-f(e^{i\theta_0})= \int_{\theta_0}^{\pi-\theta_1} \sqrt{\frac{\cos\theta_1+\cos \theta}{\cos \theta_{0} - \cos \theta}}d\theta = x,
\end{equation*}
by definition of $f$ and $\theta_{1}$,
and 
\begin{equation*}
f(1)= \int_{0}^{1} \frac{1}{s}\left( 1 - \sqrt{\frac{s^2 + 2\cos (\theta_{1}) s +1}{s^2 - 2\cos (\theta_{0}) s +1}} \right) ds = -\ell.
\end{equation*}
This completes the proof.
\end{proof}

\section{Riemann-Hilbert analysis for $x\geq x_c$}\label{section: RH}
We will now perform an asymptotic analysis of the RH problem for $Y$ as $n\to\infty$ with $x\geq x_c$. Our analysis uses the Deift/Zhou steepest descent method \cite{Deift, DKMVZ2, DKMVZ1} and shows many similarities with the analysis done in \cite{Krasovsky} for $s=0$ (or $x=+\infty$). An important difference is that we need to modify the Bessel parametrices around the points $z_0, \overline{z_0}$, see Section \ref{section: Bessel} below.

\subsection{First transformation $Y \to T$}
We define $g$ by
\begin{equation}\label{g}
g(z)=\int_{-\theta_{0}}^{\theta_{0}}\log(z-e^{i\theta})d\mu^{(\infty)}(e^{i\theta}).
\end{equation}
Below we write $\mu=\mu^{(\infty)}$, $u=u^{(\infty)}$, and $\ell=\ell^{(\infty)}$.
We have that $e^g$ is analytic for $z\in\mathbb C\setminus\gamma$, and by the Euler-Lagrange variational conditions (\ref{EL_1})-(\ref{EL_2}), we also have 
\begin{align}
&\label{g1}g_{+}(z) + g_{-}(z) - \log(z) -i\pi + \ell = 0 ,&&\mbox{ for } z\in \gamma\setminus\{z_{0},\overline{z_{0}}\},\\
&\label{g2}2g(z)  - \log(z) -i\pi+\ell<0,&&\mbox{ for }z\in S_{1}\setminus \gamma,\\
&\label{g3}g_{+}^{\prime}(z) - g_{-}^{\prime}(z)=-\frac{2\pi}{z}u(z),&&\mbox{ for }z\in \gamma\setminus\{z_{0},\overline{z_{0}}\}.
\end{align}
Define
\begin{equation}\label{def T}
T(z) = e^{\frac{n\pi i}{2}\sigma_{3}} e^{\frac{n\ell}{2}\sigma_{3}} Y(z) e^{-ng(z)\sigma_{3}} e^{-\frac{n\ell}{2}\sigma_{3}}e^{-\frac{n\pi i}{2}\sigma_{3}}.
\end{equation}
Then it is straightforward to check that $T$ satisfies the following conditions.

\subsubsection*{RH problem for $T$}
\begin{itemize}
\item[(a)] $T : \mathbb{C}\setminus S_{1} \to \mathbb{C}^{2\times 2}$ is analytic.
\item[(b)] $T$ satisfies the jump relation
\begin{equation}
T_{+}(z) = T_{-}(z) J_{T}(z),\qquad \mbox{ on } S_{1}\setminus \left\{z_{0},\overline{z_{0}}\right\},
\end{equation}
with
\begin{equation*}
J_{T}(z) = \begin{pmatrix}
\displaystyle e^{-n(g_{+}(z)-g_{-}(z))} & \displaystyle (-1)^{n}z^{-n}e^{-nV(z)}e^{n\ell}e^{n(g_{+}(z)+g_{-}(z))}e^{W(z)} \\
\displaystyle 0 & \displaystyle e^{n(g_{+}(z)-g_{-}(z))}
\end{pmatrix}.
\end{equation*}
\item[(c)] $T(z) = I + \bigO(z^{-1})$ as $z \to \infty$.
\item[(d)] As $z\to z_0=e^{i\theta_0}$ or $z\to\overline {z_0}=e^{-i\theta_0}$, we have
\begin{equation*}
\begin{array}{l l}
T(z) = \begin{pmatrix}
\bigO(1) & \bigO(\ln |z-z_0|) \\ \bigO(1) & \bigO(\ln |z-z_0|)
\end{pmatrix},  &  \mbox{ as } z \to z_0, \\ [0.8cm]
T(z) = \begin{pmatrix}
\bigO(1) & \bigO(\ln |z-\overline{z_0}|) \\ \bigO(1) & \bigO(\ln |z-\overline{z_0}|)
\end{pmatrix}, &  \mbox{ as } z \to \overline{z_0}.
\end{array}
\end{equation*}
\end{itemize}

We define 
\begin{equation}\label{def phi}
\phi(z)=2g(z)-\log z -i\pi+\ell, 
\end{equation}
so that $e^\phi$ is analytic in a neighborhood of $\gamma$ with the exception of $\gamma$ itself.
For $z\in\gamma\setminus\{z_{0},\overline{z_{0}}\}$, we have
\begin{eqnarray}\label{phi g}
\phi_+(z)&=&(g_+(z)+g_-(z)-\log z-i\pi+\ell)+(g_+(z)-g_-(z))\\
&=&g_+(z)-g_-(z),
\end{eqnarray}
by (\ref{g1}). But $g_+(z_0)=g_-(z_0)$, and integrating (\ref{g3}) between $z_0$ and $z$, and then substituting it into (\ref{phi g}), we obtain
\begin{equation}
\phi_+(z)=-2\pi\int_{z_0}^z\frac{u(\xi)}{\xi}d\xi.
\end{equation}
Analytically continuing the left and right hand side, we obtain
\begin{equation}\label{phi int}
\phi(z)=\int_{z_0}^{z} \frac{\xi+1}{\left((\xi-z_{0})(\xi-\overline{z_{0}})\right)^{1/2}} \frac{d\xi}{\xi},
\end{equation}
where the branch cut of the square root is chosen on $\gamma$. It is now convenient to express $J_T$ in terms of $\phi$:
\begin{equation*}
J_{T}(z) = \begin{cases} \begin{pmatrix}
e^{-n\phi_+(z)} & e^{W(z)} \\
0 & e^{-n\phi_-(z)}
\end{pmatrix}, & z \in \gamma\setminus\{z_{0},\overline{z_{0}}\}, \\
\begin{pmatrix}
1 & e^{-nx}e^{n\phi(z)}e^{W(z)} \\
0 & 1
\end{pmatrix}, & z \in S_{1} \setminus \gamma.
\end{cases}
\end{equation*}

\subsection{Second transformation $T \to S$}

We can factorize $J_{T}$ on $\gamma$ as follows:
\begin{multline}
\begin{pmatrix}
e^{-n\phi_+(z)} & e^{W(z)} \\
0 & e^{-n\phi_-(z)}
\end{pmatrix} = \begin{pmatrix}
1 & 0 \\ e^{-n\phi_-(z)}e^{-W(z)} & 1
\end{pmatrix}\\
\times\ \begin{pmatrix}
0 & e^{W(z)} \\ -e^{-W(z)} & 0
\end{pmatrix}\begin{pmatrix}
1 & 0 \\ e^{-n\phi_+(z)}e^{-W(z)} & 1
\end{pmatrix}.
\end{multline}
Using this factorization, we can split the jump on $\gamma$ into three different jumps on a lens-shaped oriented contour, see Figure \ref{fig_S}. It is important that the jump contour lies in the region where $W$ is analytic.
Denote by $\gamma_{+}$ and $\gamma_{-}$ the lenses around $\gamma$ on the $|z|<1$ side and the $|z|>1$ side respectively. Define
\begin{equation}\label{def S}
S(z) = \left\{ \begin{array}{l l}

T(z) \begin{pmatrix}
1 & 0 \\ -e^{-n\phi(z)}e^{-W(z)} & 1
\end{pmatrix}, &
|z|<1, z \mbox{ inside the lenses around }\gamma, \\[0.6cm]
T(z) \begin{pmatrix}
1 & 0 \\ e^{-n\phi(z)}e^{-W(z)} & 1
\end{pmatrix}, &
|z|>1, z \mbox{ inside the lenses around }\gamma, \\[0.6cm]

T(z),  &   z \mbox{ outside the lenses.}                 \\

\end{array} \right.
\end{equation}
Then $S$ solves the following RH problem.
\subsubsection*{RH problem for $S$}


\begin{figure}[t]
    \begin{center}
    \setlength{\unitlength}{1truemm}
    \begin{picture}(100,55)(-5,10)
        \cCircle(50,40){25}[f]
        \put(65,60){\thicklines\circle*{1.2}}
        \put(65,19.8){\thicklines\circle*{1.2}}
	    \qbezier(65,60)(73,40)(65,19.8)
	    \qbezier(65,60)(100,40)(65,19.8)
        \put(65,61){$z_0$}
        \put(72,36){$\gamma$}
        \put(65,36){$\gamma_+$}
        \put(78.5,36){$\gamma_-$}
        \put(21,36){$\gamma^c$}      
        \put(65,17){$\overline{z_0}$}
        \put(75,41){\thicklines\vector(0,1){.0001}}
        \put(69,41){\thicklines\vector(0,1){.0001}}
        \put(82.5,41){\thicklines\vector(0,1){.0001}}
        \put(25,39){\thicklines\vector(0,-1){.0001}}
    \end{picture}
    \caption{The jump contour for $S$.}
    \label{fig_S}
\end{center}
\end{figure}
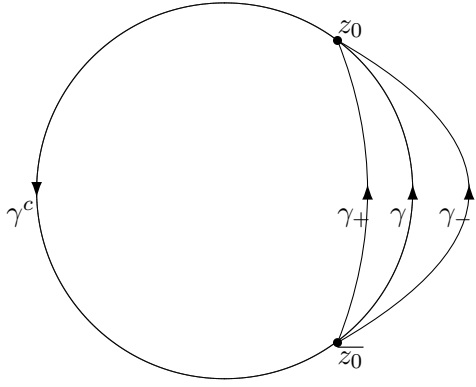

\begin{itemize}
\item[(a)] $S : \mathbb{C}\setminus (S_{1}\cup \gamma_{-} \cup \gamma_{+}) \to \mathbb{C}^{2\times 2}$ is analytic.
\item[(b)] $S$ satisfies the jump relations
\begin{equation}\label{jump S}
S_{+}(z) = S_{-}(z)J_S(z),\qquad \mbox{ for $z\in S_1\cup\gamma_-\cup\gamma_+$},
\end{equation}
where $S_+$ ($S_-$) denotes the boundary value from the left (right) of the contour, and where
\begin{equation}\label{JS}
J_S(z) = \begin{cases}\begin{pmatrix}
1 & e^{-nx}e^{n\phi(z)}e^{W(z)} \\ 0 & 1
\end{pmatrix}, & \mbox{ for } z \in S_{1}\setminus \gamma, \\
\begin{pmatrix}
0 & e^{-W(z)} \\ -e^{-W(z)} & 0
\end{pmatrix},  & \mbox{ for } z \in \gamma\setminus\{z_{0},\overline{z_{0}}\}, \\
\begin{pmatrix}
1 & 0 \\ e^{-n\phi(z)}e^{-W(z)} & 1
\end{pmatrix},  & \mbox{ for } z \in \gamma_{+}\cup\gamma_-.
\end{cases}
\end{equation} 
\item[(c)] $S(z) = I + \bigO(z^{-1})$ as $z \to \infty$.
\item[(d)] As $z\to z_0=e^{i\theta_0}$ or $z\to\overline {z_0}=e^{-i\theta_0}$, we have
\begin{equation}
S(z)=\bigO(\ln|z-z_0|),\qquad S(z)=\bigO(\ln|z-\overline{z_0}|).
\end{equation}
\end{itemize}

By (\ref{def phi}) and (\ref{g2}), we observe that the jump matrices for $S$ converge to the identity matrix on $S_{1}\setminus \gamma$ as $n\to\infty$, except at $-1$ if $x=x_c$.
On $\gamma_{+} \cup \gamma_{-}$, one shows that $\Re\phi(z)>0$ and consequently the jump matrices for $S$ also converge to the identity matrix on $\gamma_{+} \cup \gamma_{-}$.
The convergence of the jump matrices is point-wise in $z$ and breaks down as $z$ approaches $z_0, \overline{z_0}$, and also as $z$ approaches $-1$ if $x=x_c$. Therefore, we will need to construct approximations to $S$ for large $n$ in different regions of the complex plane: local parametrices will be constructed in small disks $D(z_0,r), D(\overline{z_0},r), D(-1,r)$ surrounding $z_0, \overline{z_0}, -1$, and a global parametrix will be constructed in $\mathbb C\setminus(\overline{D(z_0,r)\cup D(\overline{z_{0}},r)\cup D(-1,r)})$.

\subsection{Global parametrix}
Ignoring the exponentially small jumps for $S$ and small neighborhoods of $z_0, \overline{z_0}$, $-1$, we are led to the following RH problem:
\subsubsection*{RH problem for $P^{(\infty)}$}
\begin{itemize}
\item[(a)] $P^{(\infty)} : \mathbb{C}\setminus \gamma \to \mathbb{C}^{2\times 2}$ is analytic.
\item[(b)] $P^{(\infty)}$ has the jump
\begin{equation}\label{jump Pinfty}
P^{(\infty)}_+(z)=P^{(\infty)}_-(z)\begin{pmatrix} 0&e^{W(z)}\\-e^{-W(z)}&0
\end{pmatrix},\qquad z\in\gamma\setminus\{z_{0},\overline{z_{0}}\}.
\end{equation}
\item[(c)] $P^{(\infty)}(z) = I + \bigO(z^{-1})$ as $z \to \infty$.
\item[(d)] As $z\to z_0$ or $z\to \overline{z_0}$, we have
\begin{equation}
P^{(\infty)}(z)=\bigO(|z-z_0|^{-1/4}),\qquad P^{(\infty)}(z)=\bigO(|z-\overline{z_0}|^{-1/4}).
\end{equation}
\end{itemize}

The solution of this RHP is explicitly given by
\begin{equation}
P^{(\infty)}(z)= e^{h_\infty\sigma_3}\begin{pmatrix}
\frac{1}{2} \left( \beta(z) + \beta^{-1}(z) \right) & -\frac{1}{2i} \left( \beta(z) - \beta^{-1}(z) \right) \\
\frac{1}{2i} \left( \beta(z) - \beta^{-1}(z) \right) & \frac{1}{2} \left( \beta(z) + \beta^{-1}(z) \right)
\end{pmatrix}e^{-h(z)\sigma_3} ,
\end{equation}
where $\beta(z) = \left(\frac{z-\overline{z_{0}}}{z-z_{0}}\right)^{1/4}$ is analytic in $\mathbb{C}\setminus \gamma$ and $\beta(z) \to 1$ as $z\to\infty$, 
where $h_\infty=\lim_{z\to\infty}h(z)$, and
\begin{equation}\label{h}
h(z)=\frac{\left((z-z_0)(z-\overline{z_0})\right)^{1/2}}{2\pi i}\int_{\gamma}\frac{W(\xi)}{\left((\xi-z_0)(\xi-\overline{z_0})\right)_+^{1/2}}\frac{1}{\xi-z}d\xi,
\end{equation}
with $\left((z-z_0)(z-\overline{z_0})\right)^{1/2}$ analytic off $\gamma$. Note that $h$ is analytic in $\mathbb C\setminus \gamma$, bounded near $z_0$, $\overline{z_0}$, and $\infty$, and that it satisfies the jump relation 
\[h_+(z)+h_-(z)=W(z),\qquad z\in\gamma\setminus \{z_{0},\overline{z_{0}}\}.\] Using these properties, it is straightforward to verify that $P^{(\infty)}$ solves the above RH problem.

\subsection{Local parametrix near $z_{0}$}\label{section: Bessel}
We want to construct a function $P$ defined in $D(z_0,r)$ satisfying the following RH conditions.
\subsubsection*{RH problem for $P$}

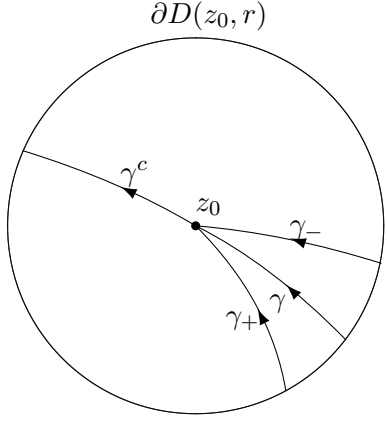
\begin{figure}[t]
    \begin{center}
    \setlength{\unitlength}{1truemm}
    \begin{picture}(100,55)(-5,10)
        \cCircle(50,40){25}[f]
        \qbezier(50,40)(40,46)(27,50)
        \qbezier(50,40)(60,30)(62,18)
        \qbezier(50,40)(60,35)(70,24.8)
        \qbezier(50,40)(62,39)(74.7,35)
        
        \put(40,45.2){\thicklines\vector(-2,1){.0001}}
        \put(58.2,29){\thicklines\vector(-1,2){.0001}}
        \put(62,32.2){\thicklines\vector(-1,1){.0001}}
        \put(62.5,38.35){\thicklines\vector(-3,1){.0001}}
        
        \put(40,46){$\gamma^{c}$}
        \put(54,27){$\gamma_{+}$}
        \put(59.8,29.2){$\gamma$}
        \put(62.5,39){$\gamma_{-}$}
        
        \put(50,40){\thicklines\circle*{1.2}}
        \put(50,42){$z_{0}$}
        \put(44,67){$\partial D(z_{0},r)$}
    \end{picture}
    \caption{The jump contour for $P$.}
    \label{figure: contour P}
\end{center}
\end{figure}

\begin{itemize}
\item[(a)] $P : D(z_{0},r) \setminus (S_{1} \cup \gamma_{+} \cup \gamma_{-}) \to \mathbb{C}^{2\times 2}$ is analytic.
\item[(b)] For $z$ on the contour shown in Figure \ref{figure: contour P}, $P$ satisfies the jump conditions
\begin{equation}\label{jumps P}
\begin{array}{l l}
P_{+}(z) = P_{-}(z) \begin{pmatrix}
0 & e^{W(z)} \\ -e^{-W(z)} & 0
\end{pmatrix}, & \mbox{ on } \gamma\setminus \{z_{0},\overline{z_{0}}\}, \\

P_{+}(z) = P_{-}(z) \begin{pmatrix}
 1 & e^{n(\phi(z)-x)}e^{W(z)} \\
 0 & 1
\end{pmatrix}, & \mbox{ on } \gamma^{c}, \\

P_{+}(z) = P_{-}(z) \begin{pmatrix}
 1 & 0  \\ e^{-n\phi(z)}e^{-W(z)} & 1
\end{pmatrix}, & \mbox{ on } \gamma_{-} \cup \gamma_{+}. \\
\end{array}
\end{equation}
\item[(c)] For $z \in \partial D (z_{0},r)$, we have
\begin{equation}\label{matching P z0} P(z) = \left(I + \bigO(n^{-1})\right) P^{(\infty)}(z),\qquad \mbox{ as $n \to \infty$.}
\end{equation}
\item[(d)] As $z$ tend to $z_{0}$, the behaviour of $P$ is
\begin{equation}\label{P local}
\begin{array}{l l}
P(z) = \bigO(\ln|z-z_0|).
\end{array}
\end{equation}
\end{itemize}

\subsubsection{Bessel model RH problem}

We will construct $P$ in terms of a model RH problem for which the solution is constructed using Bessel functions.
Consider the following model RH problem:

\subsubsection*{RH problem for $\Psi$}

\begin{itemize}
\item[(a)] $\Psi : \mathbb{C} \setminus \Sigma_{\Psi} \to \mathbb{C}^{2\times 2}$ is analytic, where $\Sigma_{\Psi} = \mathbb{R}^{-} \cup \{ xe^{\pm\frac{2\pi }{3}i} : x \in \mathbb{R}^{+} \}$, with the orientation from $\infty$ towards $0$ for the three half-lines.
\item[(b)] $\Psi$ satisfies the jump conditions
\begin{equation*}
\begin{array}{l l}
\Psi_{+}(\zeta) = \Psi_{-}(\zeta) \begin{pmatrix}
0 & 1 \\ -1 & 0
\end{pmatrix}, & \zeta \in \mathbb{R}^{-}, \\

\Psi_{+}(\zeta) = \Psi_{-}(\zeta) \begin{pmatrix}
1 & 0 \\ 1 & 1
\end{pmatrix}, & \zeta \in \{ xe^{\frac{2\pi }{3}i} : x \in \mathbb{R}^{+} \}, \\

\Psi_{+}(\zeta) = \Psi_{-}(\zeta) \begin{pmatrix}
1 & 0 \\ 1 & 1
\end{pmatrix}, & \zeta \in \{ xe^{-\frac{2\pi}{3}i} : x \in \mathbb{R}^{+} \}. \\
\end{array}
\end{equation*}
\item[(c)] $\Psi(\zeta) = \left( 2\pi \zeta^{\frac{1}{2}} \right)^{-\frac{\sigma_{3}}{2}} \frac{1}{\sqrt{2}} \begin{pmatrix}1&i\\i&1\end{pmatrix}\left(
I+\bigO (\zeta^{-\frac{1}{2}})\right) e^{2\zeta^{\frac{1}{2}}\sigma_{3}}$, \quad as $\zeta \to \infty$, $\zeta \notin \Sigma_{\Psi}$.
\item[(d)] As $\zeta$ tend to 0, the behaviour of $\Psi(\zeta)$ is
\begin{equation}\label{Psi local}
\Psi(\zeta) = \bigO (\ln |\zeta| ).
\end{equation}
\end{itemize}
This model RH problem is well-known and it can be solved explicitly using Bessel functions, see e.g.\ \cite{Kuijlaars2, Krasovsky}. The unique solution to this RH problem is given by 
\begin{equation}\label{Psi explicit}
\Psi(\zeta)=\begin{cases}
\begin{pmatrix}
I_{0}(2\zeta^{\frac{1}{2}}) & \frac{i}{\pi} K_{0}(2\zeta^{\frac{1}{2}}) \\
2\pi i \zeta^{\frac{1}{2}} I_{0}^{\prime}(2\zeta^{\frac{1}{2}}) & -2\zeta^{\frac{1}{2}} K_{0}^{\prime}(2\zeta^{\frac{1}{2}})
\end{pmatrix}, & |\arg \zeta | < \frac{2\pi}{3}, \\

\begin{pmatrix}
\frac{1}{2} H_{0}^{(1)}(2(-\zeta)^{\frac{1}{2}}) & \frac{1}{2} H_{0}^{(2)}(2(-\zeta)^{\frac{1}{2}}) \\
\pi \zeta^{\frac{1}{2}} \left( H_{0}^{(1)} \right)^{\prime} (2(-\zeta)^{\frac{1}{2}}) & \pi \zeta^{\frac{1}{2}} \left( H_{0}^{(2)} \right)^{\prime} (2(-\zeta)^{\frac{1}{2}})
\end{pmatrix}, & \frac{2\pi}{3} < \arg \zeta < \pi, \\

\begin{pmatrix}
\frac{1}{2} H_{0}^{(2)}(2(-\zeta)^{\frac{1}{2}}) & -\frac{1}{2} H_{0}^{(1)}(2(-\zeta)^{\frac{1}{2}}) \\
-\pi \zeta^{\frac{1}{2}} \left( H_{0}^{(2)} \right)^{\prime} (2(-\zeta)^{\frac{1}{2}}) & \pi \zeta^{\frac{1}{2}} \left( H_{0}^{(1)} \right)^{\prime} (2(-\zeta)^{\frac{1}{2}})
\end{pmatrix}, & -\pi < \arg \zeta < -\frac{2\pi}{3},
\end{cases}
\end{equation}
where $H_0^{(1)}$ and $H_0^{(2)}$ are the Hankel functions of the first and second kind, and $I_0$ and $K_0$ are the modified Bessel functions of the first and second kind.

\subsubsection{Modification of the model RH problem}

We will use a conformal map which maps the curves $\gamma_-, \gamma, \gamma_+$ in the vicinity of $z_0$ to (part of) the jump contour for $\Psi$. This is similar to the construction in \cite{Krasovsky}, which corresponds to the case $s=0$. If $s=0$, $P$ has no jump on $\gamma_c$ (see Figure \ref{figure: contour P}), and the model RH problem fits perfectly to construct the local parametrix $P$. In our situation however, $P$ does have a jump on $\gamma_c$, and for this reason we need to modify the model RH problem. A different but similar construction was done in \cite{BDIK}.

Define $\widehat{\Psi}$ by
\begin{equation}\label{def hatPsi}
\widehat{\Psi}(\zeta) = \left( I + A(\zeta) \right) \Psi(\zeta), 
\end{equation}
where $A$ is given by 
\begin{equation}\label{def A}
A(\zeta) = e^{-nx} F(\zeta) \begin{pmatrix}
0 & - \frac{1}{2\pi i} \ln(-\zeta) \\ 0 & 0 \\
\end{pmatrix} F^{-1}(z),
\end{equation}
with $F$ defined by
\begin{equation}\label{F}
F(\zeta) = \left\{ \begin{array}{l l}
\Psi(\zeta) \begin{pmatrix}
1 & -\frac{1}{2\pi i} \ln \zeta \\
0 & 1
\end{pmatrix} & | \arg \zeta | < \frac{2\pi}{3}, \\

\Psi(\zeta) \begin{pmatrix}
1 & 0 \\
1 & 1
\end{pmatrix} \begin{pmatrix}
1 & -\frac{1}{2\pi i} \ln \zeta \\
0 & 1
\end{pmatrix} & \frac{2\pi}{3} < \arg \zeta < \pi, \\

\Psi(z) \begin{pmatrix}
1 & 0 \\
-1 & 1
\end{pmatrix} \begin{pmatrix}
1 & -\frac{1}{2\pi i} \ln \zeta \\
0 & 1
\end{pmatrix} & -\pi < \arg \zeta < -\frac{2\pi}{3}, \\

\end{array} \right. 
\end{equation}
where in both \eqref{def A} and \eqref{F} $\ln$ has its branch cut on $\mathbb{R}^{-}$ with imaginary part between $-\pi$ and $\pi$. It is easy to check that $F$ is an entire function.
$\widehat{\Psi}$ is analytic in  $\mathbb{C} \setminus \Sigma_{\widehat{\Psi}}$, with $\Sigma_{\widehat{\Psi}}$ as shown in Figure \ref{figPsiHat}. On $\Sigma_{\widehat{\Psi}}$, it has the jump relations
\begin{equation}\label{jumps hatPsi}
\begin{array}{l l}
\widehat{\Psi}_{+}(\zeta) = \widehat{\Psi}_{-}(\zeta) \begin{pmatrix}
0 & 1 \\ -1 & 0
\end{pmatrix}, & \mbox{ on } \mathbb{R}^{-}, \\

\widehat{\Psi}_{+}(\zeta) = \widehat{\Psi}_{-}(\zeta) \begin{pmatrix}
1 & e^{-nx} \\ 0 & 1
\end{pmatrix}, & \mbox{ on } \mathbb{R}^{+}, \\

\widehat{\Psi}_{+}(\zeta) = \widehat{\Psi}_{-}(\zeta) \begin{pmatrix}
 1 & 0  \\ 1 & 1
\end{pmatrix}, & \mbox{ on } \left\{ xe^{\frac{2\pi}{3}i} : x \in \mathbb{R}^{+} \right\}, \\

\widehat{\Psi}_{+}(\zeta) = \widehat{\Psi}_{-}(\zeta) \begin{pmatrix}
 1 & 0 \\ 1 & 1
\end{pmatrix}, & \mbox{ on } \left\{ xe^{-\frac{2\pi}{3}i} : x \in \mathbb{R}^{+} \right\}.
\end{array}
\end{equation}


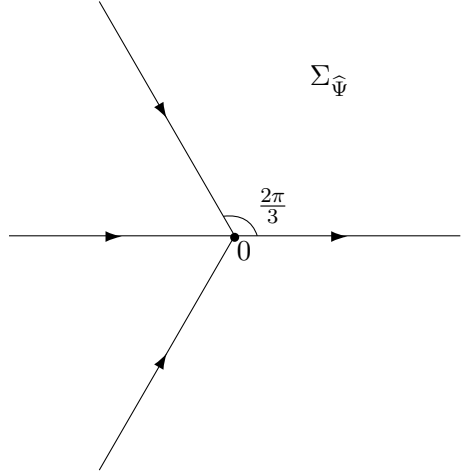
\begin{figure}[t]
    \begin{center}
    \setlength{\unitlength}{1truemm}
    \begin{picture}(100,55)(-5,10)
        \put(50,40){\line(1,0){30}}
        \put(50,40){\line(-1,0){30}}
        \put(50,39.8){\thicklines\circle*{1.2}}
        \put(50,40){\line(-0.5,0.866){18}}
        \put(50,40){\line(-0.5,-0.866){18}}
        \qbezier(53,40)(52,43)(48.5,42.598)
        \put(53,43){$\frac{2\pi}{3}$}
        \put(50.3,36.8){$0$}
        \put(65,39.9){\thicklines\vector(1,0){.0001}}
        \put(35,39.9){\thicklines\vector(1,0){.0001}}
        \put(41,55.588){\thicklines\vector(0.5,-0.866){.0001}}
        \put(41,24.412){\thicklines\vector(0.5,0.866){.0001}}
        \put(60,60){$\Sigma_{\widehat{\Psi}}$}
    \end{picture}
    \caption{The jump contour for $\widehat{\Psi}$.}
    \label{figPsiHat}
\end{center}
\end{figure}

\subsubsection{Construction of the local parametrix}

We construct $P$ as follows,
\begin{equation}\label{Explicit_P}
P(z)=E(z)\widehat\Psi(n^2\zeta(z))e^{-\frac{n}{2}\phi(z)\sigma_3}e^{-\frac{1}{2}W(z)\sigma_{3}},
\end{equation}
where $E$ is an analytic function in $D(z_0,r)$, and where
$\zeta(z) = \frac{1}{16} \phi(z)^{2}$. By (\ref{phi int}), we have that $\zeta(z)$ is a conformal map near $z_0$, and that $\zeta(z_0)=0$. Moreover, $\zeta$ maps $\gamma\cap D(z_0,r)$ to part of the real line. We now fix the lens-shaped contours $\gamma_-$ and $\gamma_+$ by requiring that $\zeta(\gamma_-\cup\gamma_+)\subset \Sigma_{\widehat\Psi}$.

For any analytic function $E$, we have that $P$ defined in \eqref{Explicit_P} satisfies conditions (a), (b), and (d) of the RH problem for $P$. Indeed, $P$ is analytic in $D(z_0,r)\setminus(\gamma\cup\gamma_+\cup\gamma_-\cup\gamma^c)$ by construction, and 
by (\ref{Explicit_P}) and (\ref{jumps hatPsi}), it is straightforward to verify that the jump condition
(\ref{jumps P}) holds. The logarithmic behavior (\ref{P local}) of $P$ near $z_0$ follows from the logarithmic behavior (\ref{Psi local}) together with the definition (\ref{def hatPsi})-(\ref{def A}) of $\widehat\Psi$.
If we define $E(z)$ by 
\begin{equation*}
E(z) = P^{(\infty)}(z)e^{\frac{1}{2}W(z)\sigma_{3}}  \frac{1}{\sqrt{2}} \begin{pmatrix}
1 & -i \\ -i & 1
\end{pmatrix} \left( \frac{1}{2} n\pi \phi(z) \right)^{\frac{\sigma_{3}}{2}},
\end{equation*}
we have in addition the matching condition (\ref{matching P z0}) for $P$.  Using the jump relation for $P^{(\infty)}$, it is easily verified that $E$ is analytic near $z_0$.
This completes the construction of the local parametrix near $z_0$.

\subsection{Local parametrix near $\overline{z_0}$}

Once the local parametrix near $z_0$ constructed, the local parametrix near $\overline{z_0}$ is easy to construct. Define for $z \in D(\overline{z_{0}},r)$,  $P(z)=\overline{P(\overline{z})}$, where $\overline{P(\overline{z})}$ refers to the local parametrix constructed near $z_{0}$. Then, $P$ satisfies the following RH conditions.

\subsubsection*{RH problem for $P$}

\begin{itemize}
\item[(a)] $P : D(\overline{z_{0}},r) \setminus (S_{1} \cup \gamma_{+} \cup \gamma_{-}) \to \mathbb{C}^{2\times 2}$ is analytic.
\item[(b)] For $z\in D(\overline{z_{0}},r)$ and $z$ on the jump contour, $P$ satisfies the jump conditions
\begin{equation*}
\begin{array}{l l}
P_{+}(z) = P_{-}(z) \begin{pmatrix}
0 & e^{W(z)} \\ -e^{-W(z)} & 0
\end{pmatrix}, & \mbox{ on } \gamma\setminus\{z_{0},\overline{z_{0}}\}, \\

P_{+}(z) = P_{-}(z) \begin{pmatrix}
 1 & e^{n(\phi(z)-x)}e^{W(z)} \\
 0 & 1
\end{pmatrix}, & \mbox{ on } \gamma^{c}, \\

P_{+}(z) = P_{-}(z) \begin{pmatrix}
 1 & 0  \\ e^{-n\phi(z)}e^{-W(z)} & 1
\end{pmatrix}, & \mbox{ on } \gamma_{-} \cup \gamma_{+}. \\
\end{array}
\end{equation*}
\item[(c)] For $ z \in \partial D (\overline{z_{0}},r)$, we have
\begin{equation}\label{matching P z0bar} P(z) = \left(I + \bigO(n^{-1})\right) P^{(\infty)}(z),\qquad\mbox{ as $n \to \infty$.}
\end{equation}
\item[(d)] As $z$ tend to $\overline{z_{0}}$, the behaviour of $P$ is
\begin{equation*}
\begin{array}{l l}
P(z) = \bigO(\ln|z-\overline{z_0}|).
\end{array}
\end{equation*}
\end{itemize}

\subsection{Local parametrix near $-1$}
For $x>x_c+\delta$, the jump matrix for $S$ converges exponentially fast to the identity matrix as $n\to\infty$ near  $-1$. However, as $x$ approaches $x_c$, the convergence becomes slower, and for $x=x_c$, the jump matrix for $S$ does not converge to $I$ any longer. Therefore, we need to construct a local parametrix near $-1$. This construction can be done for any $x\geq x_c$ but is only necessary when $x$ is close to $x_c$. The local parametrix should satisfy the following conditions.
\subsubsection*{RH problem for $P$}
\begin{itemize}
\item[(a)] $P : D(-1,r)\setminus S_{1} \to \mathbb{C}^{2\times 2}$ is analytic.
\item[(b)] $P$ has the jump
\begin{equation}
P_+(z)=P_-(z)\begin{pmatrix}1 & e^{n(\phi(z)-x)}e^{W(z)}\\0 &1
\end{pmatrix}, \qquad z\in S_{1}\cap D(-1,r).
\end{equation}
\item[(c)] For $z\in\partial D(-1,r)$, we have
\begin{equation}
\label{matching P -1}P(z) = (I + o(1))P^{(\infty)}(z),\qquad \mbox{ as $n \to \infty$.}
\end{equation}
\end{itemize}

The solution of this RHP is given by

\begin{equation}\label{P-1}
P(z) = P^{(\infty)}(z) \begin{pmatrix}
1 & \tilde{h}(z) \\ 0 & 1
\end{pmatrix},
\end{equation}
with $\displaystyle \tilde{h}(z) = \frac{1}{2\pi i} \int_{S_{1}\cap D(-1,r)} \frac{e^{n(\phi(s)-x)}e^{W(s)}}{s-z}ds$.
Using the fact that $\phi(s)-x_c$ has a double zero at $s=1$, it is straightforward to verify that
\begin{equation}
\tilde{h}(z)=\bigO(n^{-1/2}e^{n(x_c-x)}),\qquad\mbox{for $z\in \partial D(-1,r)$, as $n\to\infty$,}
\end{equation}
and this implies the matching condition (\ref{matching P -1}).

Note that we use the same notation $P$ for the different local parametrices defined in $D(z_0,r)$, $D(\overline{z_0,r})$, and $D(-1,r)$.

\subsection{Final transformation $S\mapsto R$}

Define
\begin{equation}\label{R}
R(z)=\begin{cases}
S(z)P^{(\infty)}(z)^{-1},&z\in\mathbb C\setminus(\overline{D(z_0,r)}\cup\overline{D(\overline{z_0},r)}\cup\overline{D(-1,r)}),\\
S(z)P(z)^{-1}, & z\in D(z_0,r)\cup D(\overline{z_0},r)\cup D(-1,r).
\end{cases}
\end{equation}
$P$ was constructed in such a way that it has exactly the same jump relations as $S$ in $D(z_0,r)\cup D(\overline{z_0},r)\cup D(-1,r)$, and as a consequence $R$ has no jumps at all inside those disks. Moreover, from the local behaviour of $S$ and $P$ near $z_0$, $\overline{z_0}$ and $-1$, it follows that $R$ is analytic at these three points.
We have the following RH problem for $R$:
\subsubsection*{RH problem for $R$}


\begin{figure}[t]
    \begin{center}
    \setlength{\unitlength}{1truemm}
    \begin{picture}(100,55)(-5,10)
        \cCircle(65,60){7.5}[f]
        \cCircle(65,20){7.5}[f]
        \cCircle(25,40){7.5}[f]
        \qbezier(69.5,54)(75,40)(69.5,26)
        \qbezier(72.3861,58.6976)(100,40)(72.3861,21.3024)
        \qbezier(26.125,47.415)(35.9673,69.0665)(58.3936,63.5488)
        \qbezier(26.125,32.585)(35.9673,10.9335)(58.3936,16.4512)
        
        \put(67,67.3){\thicklines\vector(1,0){.0001}}
        \put(63.8,12.6){\thicklines\vector(-1,0){.0001}}
        \put(17.5,41.5){\thicklines\vector(0,1){.0001}}
        \put(72.3,41){\thicklines\vector(0,1){.0001}}
        \put(86.2,42){\thicklines\vector(0,1){.0001}}
        \put(50,64.5){\thicklines\vector(-1,0){.0001}}
        \put(51,15.3){\thicklines\vector(1,0){.0001}}
        
        
        \put(64,59){$z_{0}$} 
        \put(64,19){$\overline{z_{0}}$}
        \put(22.5,38.8){$-1$}
        
        \put(48,68){$\Sigma_{R}$}
    \end{picture}
    \caption{The jump contour for $R$.}
    \label{figure: contour R}
\end{center}
\end{figure}
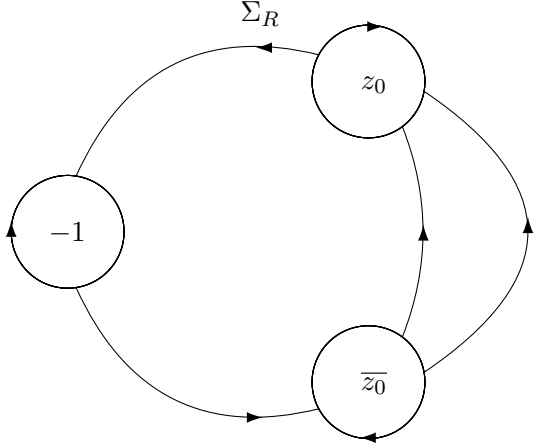

\begin{itemize}
\item[(a)] $R : \mathbb{C} \setminus \Sigma_{R} \to \mathbb{C}^{2\times 2}$ is analytic, with $\Sigma_R$ as in Figure \ref{figure: contour R}.
\item[(b)] $R$ satisfies the jump conditions:
\begin{equation*}
\begin{array}{l l}
\displaystyle R_{+}(z) = R_{-}(z) \left(I + \bigO(n^{-1})\right), & \mbox{ for } z \in \partial D(z_{0},r) \cup \partial D(\overline{z_{0}},r), \\[0.2cm]

\displaystyle R_{+}(z) = R_{-}(z) \left(I + \bigO(n^{-\frac{1}{2}}e^{n(x_c-x)})\right), & \mbox{ for } z \in \partial D(-1,r), \\[0.2cm]

\displaystyle R_{+}(z) = R_{-}(z) (I+\bigO(e^{-cn})), c>0 & \mbox{ for } z \mbox{ elsewhere on } \Sigma_{R}.

\end{array}
\end{equation*}
\item[(c)] $R(z) = I + \bigO(z^{-1})$ as $z \to \infty$.
\end{itemize}
As $n\to\infty$ with $s\leq e^{-x_cn}$, it follows from the standard theory for small-norm RH problems that \begin{equation}R(z) = I + \bigO(n^{-\frac{1}{2}}e^{n(x_c-x)})+\bigO(n^{-1}),\qquad R^{\prime}(z) = \bigO(n^{-\frac{1}{2}}e^{n(x_c-x)})+\bigO(n^{-1}),\label{as R}
\end{equation}
uniformly for $z \in \mathbb{C}\setminus \Sigma_{R}$.
By a more detailed analysis of Cauchy operators associated to the RH problem for $R$, one obtains in addition that
\begin{equation}\label{dxR}
\partial_x R(z)=\bigO(n^{1/2}e^{n(x_c-x)}).
\end{equation}
The latter estimate is not really needed for the proof of Theorem \ref{theorem: extensionWidom}, but will be needed in our alternative proof if $W=0$, see Section \ref{section: proof W0}.
Since $s=e^{-xn}$, it follows that 
\begin{equation}\label{dsR}
\partial_s R(z)=-\frac{e^{xn}}{n}\partial_x R(z)=\bigO(n^{-1/2}e^{nx_c}).
\end{equation}

\section{Asymptotics for $D_n(s,\theta_0,W)$}\label{section: int}

We will now complete the proof of Theorem \ref{theorem: extensionWidom} using the results from the previous sections. We first do this for general analytic $W$ in the case where $\epsilon<\theta_0<\pi-\epsilon$ for some $\epsilon>0$. Afterwards we will explain how the proof can be somewhat simplified if $W=0$, and we will show how it can be extended to the case where $\theta_0$ approaches $\pi$ at a sufficiently slow rate.

\subsection{Proof of Theorem \ref{theorem: extensionWidom} for $\epsilon<\theta_0<\pi-\epsilon$}

By inverting the transformations $Y\mapsto T\mapsto S\mapsto R$, we obtain an expression for $R$ in terms of $Y$, involving the parametrices $P$ and $P^{(\infty)}$.
For $z\in \gamma^c\cap (D(z_0,r)\cup D(\overline{z_0},r)\cup D(-1,r))$, we have
\begin{equation}\label{YinRin}
Y_\pm(z) = e^{-\frac{n\pi i}{2}\sigma_{3}}e^{-\frac{n\ell}{2}\sigma_3 } R(z) P_\pm(z) e^{ng(z)\sigma_3}e^{\frac{n\ell}{2} \sigma_{3}}e^{\frac{n\pi i}{2}\sigma_{3}},
\end{equation}
and for $z\in\gamma^c\setminus (D(z_0,r)\cup D(\overline{z_0},r)\cup D(-1,r))$,
\begin{equation}\label{YinR out}
Y_\pm(z) = e^{-\frac{n\pi i}{2}\sigma_{3}}e^{-\frac{n\ell}{2} \sigma_3} R_\pm(z) P^{(\infty)}(z) e^{ng(z)\sigma_3}e^{\frac{n\ell}{2} \sigma_{3}}e^{\frac{n\pi i}{2}\sigma_{3}}.
\end{equation}
It follows that
\begin{equation}\label{Yindiffid}
\left[Y^{-1}(z)Y'(z)\right]_{21}=\begin{cases} (-1)^{n}e^{2ng(z)}e^{n\ell}A_1(z),& z\in \gamma^c\cap (D(z_0,r)\cup D(\overline{z_0},r)\cup D(-1,r)),\\ (-1)^{n}e^{2ng(z)}e^{n\ell}A_2(z),&z\in\gamma^c\setminus (D(z_0,r)\cup D(\overline{z_0},r)\cup D(-1,r)),
\end{cases}
\end{equation}
where
\begin{align}
&A_{1}(z) = \left[ P^{-1}(z) R^{-1}(z) R^{\prime}(z) P(z) + P^{-1}(z)P^{\prime}(z) \right]_{21}, \\
&A_{2}(z) = \left[ {P^{(\infty )}}^{-1}(z) R^{-1}(z) R^{\prime}(z) P^{(\infty )}(z) + {P^{(\infty )}}^{-1}(z){P^{(\infty )}}^{\prime}(z) \right]_{21}.
\end{align}
Note that the boundary values of $A_1(z)$ and $A_2(z)$ as $z$ is approached from the inside and the outside of the unit circle are the same.
For $z\in D(-1,r)$, the local parametrix $P$ is given by (\ref{P-1}), and one verifies that the formulas for $A_1$ and $A_2$ coincide in this case.

 Substituting (\ref{Yindiffid}) in the differential identity \eqref{diffid1general}, we find 
\begin{equation} \label{diff_identity_W}
\partial_{s} \ln D_{n} (s,\theta_{0},W) = I_{1} + I_{2},
\end{equation}
where
\begin{equation*}
\begin{array}{l}
\displaystyle I_{1} = \frac{1}{2\pi i} \int_{\gamma_{1}^{c}} z^{-n} e^{2ng(z)}e^{n\ell} A_{1}(z) e^{W(z)}dz, \\[0.4cm]
\displaystyle I_{2} = \frac{1}{2\pi i} \int_{\gamma_{2}^{c}} z^{-n} e^{2ng(z)} e^{n\ell} A_{2}(z) e^{W(z)}dz, 
\end{array}
\end{equation*}
with $\gamma_{1}^{c} = \gamma^{c} \cap (D(z_{0},r) \cup D(\overline{z_{0}},r))$ and $\gamma_{2}^{c} = \gamma^{c}\setminus (D(z_{0},r) \cup D(\overline{z_{0}},r))$.

Now we need to know how $I_1$ and $I_2$ behave for large $n$ and $s\leq e^{-x_cn}$. For $A_2$, we note that $P^{(\infty)}$ does not depend on $n$, and that $R$ and $R'$ are uniformly bounded by (\ref{as R}). This implies that $A_{2}(z)$ is uniformly bounded on $\gamma_2^c$ for large $n$. For $A_1$, we need to take a closer look at the construction of $P$ near $z_0$ and $\overline{z_0}$, but it is straightforward to show that $P(z)=\bigO(n)$, $P^{-1}(z)=\bigO(n)$, and $P'(z)=\bigO(n^2)$ for $z\in\gamma_1^c$.
We get
\begin{equation*}
\begin{array}{r c l}
\displaystyle |I_{1}| & = & \displaystyle \bigO \left(n^3 \int_{[\theta_{0},\theta_{0}+r] \cup [2\pi - \theta_{0}-r,2\pi - \theta_{0}]} \left| e^{2ng(e^{i\alpha})+n\ell} \right| d\alpha \right), \\[0.4cm]
\displaystyle |I_{2}| & = & \displaystyle \bigO \left( \int_{\theta_{0}+r}^{2\pi -\theta_{0}-r} \left| e^{2ng(e^{i\alpha})+n\ell} \right| d\alpha \right),
\end{array}
\end{equation*}
as $n\to\infty$ with $s\leq e^{-x_c n}$.
The function $2g(z)+\ell-x_{c}$ is always negative on $\gamma^{c}$ except that it has a zero of order two at $z=-1$. 
Therefore we obtain after a straightforward analysis that
\begin{align}
&|I_{1}| = \bigO\left(e^{n(x_{c}-C)} \right),\qquad C>0, \label{Estimate1}\\
&|I_{2}| = \bigO (n^{-1/2}e^{nx_{c}}), \label{Estimate2}
\end{align}
as $n\to\infty$, $s\leq e^{-x_c n}$.

If we integrate \eqref{diff_identity_W} from $0$ up to $s=e^{-xn} \leq e^{-x_{c}n}$, we finally obtain the desired estimate
\begin{equation*}
\ln D_{n} (s,\theta_{0},W) - \ln D_{n}(0,\theta_{0},W) = \int_{0}^{s} (I_{1} + I_{2}) ds^{\prime} =\bigO(n^{-1/2}e^{-n(x-x_c)}).
\end{equation*}

\subsection{Alternative proof of Theorem \ref{theorem: extensionWidom} if $W=0$, $s=o(n^{1/2}e^{-nx_c})$}\label{section: proof W0}
In the case where $W=0$, there is an alternative approach to prove Theorem \ref{theorem: extensionWidom}: we can use the differential identity (\ref{D_n_in_term_of_Y}) instead of (\ref{diffid1general}). This does not make the proof much shorter, but it has the advantage that no integrals have to be estimated, and that we only need information about $Y$ at the points $0$ and $z_0$, instead of on the entire curve $\gamma_c$. The objects in the differential identity can be computed more explicitly in this case by the following result.

\begin{proposition}\label{Y11 and Y12}
Let $W=0$. As $n\to\infty$ with $x>x_c$, we have
\begin{align} \label{Asymptotics_Y12_Y11}
& Y_{12}(0)=  e^{-n\ell}\left[ \sin \frac{\theta_{0}}{2}+\bigO (n^{-1/2} )\right],\\
&\label{dsY} \partial_s \ln Y_{12}(0)= \bigO (n^{-1/2}e^{nx_c} ),\\
&Y_{11}(z_0)=e^{-\frac{n\ell}{2}}\bigO (n^{1/2}).
\end{align}
\end{proposition}
\begin{proof}
Using (\ref{YinR out}) and the expressions $g(0) = \pi i$ and $P_{12}^{(\infty)}(0) = \sin \frac{\theta_{0}}{2}$ (if $W=0$), we get the result for $Y_{12}(0)$.

For $\partial_s \ln Y_{12}(0)$, we have
\begin{equation}
\partial_s \ln Y_{12}(0)=\frac{\partial_s\left(R_{11}(0)P_{12}^{(\infty)}(0)+R_{12}(0)P^{(\infty)}_{22}(0)\right)}{R_{11}(0)P_{12}^{(\infty)}(0)+R_{12}(0)P^{(\infty)}_{22}(0)}.
\end{equation}
By (\ref{dsR}), this yields (\ref{dsY}).

For the rest of this proof, we assume that $|z|<1$ and that $z$ lies outside of the lenses and in $D(z_0,r)$. Then we have by (\ref{YinRin}),
\begin{equation}\label{Y11}
Y_{11}(z) = e^{ng(z)} \left[ R_{11}(z)P_{11}(z) + R_{12}(z)P_{21}(z) \right].
\end{equation}
By (\ref{g}) and (\ref{g1}), we can show that
\begin{equation*}
g(z_{0}) = -\frac{\ell}{2} + i \frac{\theta_{0}+\pi}{2}.
\end{equation*}
On the other hand, by (\ref{Psi explicit}), as $z \to z_{0}$ for fixed $n$, we have
\begin{equation*}
\Psi_{11}(n^{2}\zeta(z)) = 1+ \bigO(z-z_{0}), \quad \Psi_{21}(n^{2}\zeta(z)) = \bigO(z-z_{0}).
\end{equation*}
This implies, by (\ref{def hatPsi}), that
\begin{equation*}
P_{j1}(z_{0}) = E_{j1}(z_{0})(1+\bigO(e^{-nx}))e^{-\frac{n}{2}\phi(z_{0})},\qquad\mbox{ as $n\to\infty$}.
\end{equation*}
Since $\phi(z_{0}) = 0$, we have $P_{j1}(z_{0}) = \bigO(\sqrt{n})$, $j=1,2$, and
\begin{equation*}
Y_{11}(z_{0}) = e^{n \left( -\frac{\ell}{2} +i \frac{\theta_{0}+\pi}{2} \right)}  \left( P_{11}(z_{0}) + \bigO \left( 1 \right) \right)=e^{-\frac{n\ell}{2}}\bigO (n^{1/2}),
\end{equation*}
as $n\to\infty$.
\end{proof}

By Proposition \ref{Y11 and Y12} and (\ref{D_n_in_term_of_Y}), we have
\begin{align}
&n \partial_{s}\ln Y_{12}(0) = \bigO(n^{1/2}e^{nx_c}),\\
& \frac{2(1-s)}{\pi} \Im \left( \frac{\overline{Y_{11}(e^{i\theta_{0}})}}{\sqrt{Y_{12}(0)}} \partial_{s} \left( \frac{Y_{11}(e^{i\theta_{0}})}{\sqrt{Y_{12}(0)}} \right) \right) = \bigO(n),
\end{align}
as $n\to\infty$.
Therefore, using \eqref{integration_diff_identity}, we get
\begin{equation*}
\begin{array}{r c l}
\displaystyle \ln D_{n}(s,\theta_{0},0) & = & \displaystyle \ln D_{n}(0,\theta_{0},0) + \int_{0}^{s} \bigO(n^{1/2}e^{nx_c}) ds^{\prime} \\
& = & \displaystyle \ln D_{n}(0,\theta_{0},0) + \bigO(n^{1/2}e^{-n(x-x_{c})}),
\end{array}
\end{equation*} and we rederive Theorem \ref{theorem: extensionWidom} with a slightly worse error term which is only small if $s=o(n^{1/2}e^{-nx_c})$.

\subsection{Extension to the case where $\theta_0$ depends on $n$}
The RH analysis done in the previous section is valid as $n\to\infty$ with $\epsilon<\theta_0<\pi-\epsilon$ and $\epsilon>0$ independent of $n$. If $\theta_0=\theta_0(n)$ depends on $n$ and approaches $\pi$ as $n\to\infty$, the arc $\gamma$ grows and the {\em gap} $(e^{i\theta_0},e^{i(2\pi-\theta_0)})$ closes. 
Similarly to \cite{Krasovsky}, we will show here that the RH analysis carried out before remains valid as long as $n(\pi-\theta_0)$ is large.

A first problem in the RH analysis is that we need to allow the radius $r=r(n)$ of the disks $D(z_0,r)$, $D(-1,r)$, and $D(\overline{z_0},r)$ to depend on $n$ (or on $\theta_0$) in order to prevent the disks to overlap.
For instance, we can keep the disks separated from each other if we let $r(\theta_0)=\delta(\pi-\theta_0)$, with $\delta>0$ a sufficiently small but fixed number.
Then the local parametrices near $z_0, \overline{z_0}, -1$ can be constructed in exactly the same way as before. However, in order to make the asymptotic analysis work, it is crucial that the jump matrix for $R$ tend to $I$ uniformly as $n\to\infty$ on the $n$-dependent jump contour. Recall that the jump matrix for $R$ is given by
\begin{equation}
J_R(z)=\begin{cases}P(z)P^{(\infty)}(z)^{-1},&z\in\partial D(z_0,r)\cup \partial D(\overline{z_0},r)\cup \partial D(-1,r),\\
P^{(\infty)}(z)J_S(z) P^{(\infty)}(z)^{-1},&z\in\Sigma_R\setminus\left(\partial D(z_0,r)\cup \partial D(\overline{z_0},r)\cup \partial D(-1,r)\right),
\end{cases}
\end{equation}
where $J_S$ denotes the jump matrix for $S$ given in (\ref{JS}).

The first thing to notice is that  $h(z)$, defined in \eqref{h}, and $\beta(z)$, defined right before (\ref{h}), are uniformly bounded on the jump contour $\Sigma_{R}$. This implies that $P^{(\infty)}(z)$ is uniformly bounded for $z\in\Sigma_R$. 
Next, by (\ref{phi int}), we have \begin{equation}|\phi(z)| \geq C \sin(\pi-\theta_{0}),\qquad z\in\Sigma_R,
\end{equation} 
for a constant $C > 0$ independent of $z$ and $\theta_{0}$. 
By (\ref{JS}), this already implies that the jump matrix $J_R$ is $I+\bigO(n^{-1}(\pi-\theta_0)^{-1})$ on $\Sigma_R$, except on the three circles around $z_0, \overline{z_0}$, and $-1$.
On $\partial D(z_0,r)$ and $\partial D(\overline{z_0},r)$, we have
\begin{eqnarray}
J_R(z)&=&P(z)P^{(\infty)}(z)^{-1}=P^{(\infty)}(z)\left(I+\bigO(n^{-1}\phi(z)^{-1})\right)P^{(\infty)}(z)^{-1}\\&=&I+\bigO(n^{-1}(\pi-\theta)^{-1}).
\end{eqnarray}
Finally, on $\partial D(-1,r)$, by (\ref{P-1}), we have
\begin{equation}
J_R(z)=P(z)P^{(\infty)}(z)^{-1}=I+\bigO(\widetilde{h}(z))=I+\bigO(n^{-1/2}e^{-n(x-x_c)}).
\end{equation}
It follows from these considerations that \[J_{R}(z) - I = \bigO(n^{-1}(\pi-\theta_0)^{-1}) + \bigO(n^{-1/2}e^{-n(x-x_c)}),\qquad\mbox{ as $n\to\infty$},\] with $\theta_0$ approaching $\pi$ sufficiently slowly such that $n (\pi-\theta_{0}) \to + \infty$. By the small-norm theory for RH problems, it follows that \begin{equation}
R(z) = I+\bigO \left( n^{-1}(\pi-\theta_{0})^{-1}\right) + \bigO(n^{-1/2}e^{-n(x-x_c)}),
\end{equation} and \begin{equation}
\qquad R^{\prime}(z) = \bigO \left( n^{-1}(\pi-\theta_{0})^{-1}\right) + \bigO(n^{-1/2}e^{-n(x-x_c)}),
\end{equation}
as $n\to\infty$, $n (\pi-\theta_{0}) \to + \infty$.
The proof of Theorem \ref{theorem: extensionWidom} can now be completed in the same way as before, where the estimate \eqref{Estimate1} remains valid, \eqref{Estimate2} becomes
\begin{equation*}
|I_{2}| =\bigO((\pi-\theta_0)^{1/2}n^{-1/2}e^{nx_c}),
\end{equation*}
and we obtain the error term (\ref{error large y}).

\subsection{Heuristic discussion of the asymptotics for $s>e^{-x_cn}$}\label{section: elliptic}

As $n\to\infty$ with $s=e^{-xn}>e^{-(x_c-\delta)n}$, $\delta>0$, it is expected that the behavior of the Toeplitz determinants $D_n(s,\theta_0,W)$ is very different from the behavior described in Theorem \ref{theorem: extensionWidom}. As shown in Proposition \ref{thm200}, the equilibrium measure is then supported on two disjoint arcs. In the RH analysis, that means that the RH problem for the global parametrix $P^{(\infty)}$ will become more involved. It is well understood that $P^{(\infty)}$ can then be constructed using elliptic $\theta$-functions, see e.g.\ \cite{Bleher, DKMVZ2, DKMVZ1}. For this reason and because of the analogy with the Fredholm determinants of the sine kernel discussed in Section \ref{section: Fredholm}, we expect that the asymptotic behavior of the Toeplitz determinants will also involve elliptic $\theta$-functions, and will be oscillatory. A transition between the Widom asymptotics (\ref{Asymptotic_one_arc}) and the Fisher-Hartwig asymptotics (\ref{Asymptotic_FH}) should become visible here.
A rigorous analysis is delicate, and we plan to come back to this in a future publication.

\section*{Acknowledgements}
The authors are grateful to Igor Krasovsky for useful discussions and for sharing an early version of \cite{BDIK}, and to Alexander Bufetov for useful discussions. They were supported by the Belgian Interuniversity Attraction Pole P07/18 and by F.R.S.-F.N.R.S. TC was also supported by the European Research Council under the European Union's Seventh Framework Programme (FP/2007/2013)/ ERC Grant Agreement n.\, 307074.

\end{document}